\documentclass[reqno, 12pt, a4paper]{amsart}
\usepackage[dvipsnames]{xcolor}
\usepackage{tikz-cd}
\usepackage[utf8]{inputenc}
\usepackage[T1]{fontenc}
\usepackage{amsmath}
\usepackage{amsfonts,amssymb}

\usepackage{tabularx}
\usepackage[dvipsnames]{xcolor}

\definecolor{gruen_melone}{HTML}{006f4c} 
\definecolor{rot_melone}{HTML}{a50b1b}
\colorlet{rot_melone}{rot_melone!80!gray}

\usepackage{hyperref}
\hypersetup{
  colorlinks   = true, 
  urlcolor     = gruen_melone, 
  linkcolor    = gruen_melone, 
  citecolor   = rot_melone 
}

\usepackage{tikz-cd}
\usepackage[utf8]{inputenc}
\usepackage[T1]{fontenc}
\usepackage{amsmath}
\usepackage{amsfonts,amssymb}
\usetikzlibrary{matrix}
\usepackage{tensor}
\newcommand{\swap}[2]{ \tensor[_#1]{\sharp}{_#2}}
\usepackage{bbm}

\newcommand{\psdm}{\textsc{psdm}}
\newcommand{\dip}{\textsc{divp}}
\newcommand{\Msc}{\mathscr{M}}
\usepackage{algorithm2e}
\graphicspath{{./pics_prostyB}}

\newcommand{\E}{\mathbb{ E}  }

\makeatletter
\newcommand{\vast}{\bBigg@{4.030}}
\newcommand{\vastplus}{\bBigg@{7.30}}
\newcommand{\Vast}{\bBigg@{18.30}}
\makeatother

\usepackage{tabularx}
\newcolumntype{L}{>{\arraybackslash}m{12cm}}
\let\langleb=\langle
\let\rangleb=\rangle
\usepackage[cal=euler, scr=boondoxo]{mathalfa}

\def\[#1\]{%
  \begin{align}#1%
  \end{align}%
}

 \usepackage{mathabx}
\let\langle=\langleb
\let\rangle=\rangleb

\newcommand{\gris}[1]{\textcolor{gray}{#1}}

\definecolor{azulESI}{HTML}{1266AE}
\usepackage{kbordermatrix}
\usepackage{blkarray}
\definecolor{AZULESI}{HTML}{1266AE}
\usepackage{enumerate}

\usepackage{nicematrix}
\usepackage{soul}
\usepackage[utf8]{inputenc}

\newcommand{\wenn}{\text { if }}

\newcommand{\summ}[1]{\sum_{#1=1}^m }
\newcommand{\dcupa}{\dot\cup_{a=1}}

\makeatletter
\ifcase \@ptsize \relax
  \newcommand{\nano}{\@setfontsize\miniscule{3.5}{4.5}}
\or
  \newcommand{\nano}{\@setfontsize\miniscule{4.5}{5.5}}%
\or
  \newcommand{\nano}{\@setfontsize\miniscule{4.5}{5.5}}%
\fi
\makeatother

\newcommand{\balita}{\raisebox{1.3pt}{\text{\nano$\bullet$\hspace{.75pt}}}}

\usepackage{tikz}
\usetikzlibrary{automata,arrows,topaths}

\newcommand{\inter}{^{\text{\tiny int}}}

\newcommand{\inv}{^{-1}}

\newcommand{\itemb}{\item[\balita]}
\newcommand{\cH}{\mathcal{H}}

\newcommand{\can}{\mathrm{canon}}

\usepackage{float}
\usepackage[font=small,labelfont=sc]{caption}
\usepackage{subcaption}
\usepackage{amsthm}
\numberwithin{equation}{section}
\newtheoremstyle{mytheoremstyle} 
    {10pt}                    
    {8pt}                    
    {\itshape}                   
    {}                           
    {\scshape}                   
    {.}                          
    {.5em}                       
    {}  

\usepackage[includeheadfoot,margin=2.2cm]{geometry}
\makeatletter
\newcommand{\leqnomode}{\tagsleft@true}
\newcommand{\reqnomode}{\tagsleft@false}
\makeatother
\theoremstyle{mytheoremstyle}

\newtheorem*{theorem*}{Theorem A}
\newtheorem{theorem}{Theorem}[section]

\newtheorem{corollary}[theorem]{Corollary}
 \newtheorem{lemma}[theorem]{Lemma}
 \newtheorem{claim}[theorem]{Claim}
 \newtheorem{proposition}[theorem]{Proposition}
 \newtheoremstyle{definition} 
    {8pt}                    
    {5pt}                    
    {}                   
    {}                           
    {\scshape}                   
    {.}                          
    {.5em}                       
    {}  

 \theoremstyle{definition}
 \newtheorem{definition}[theorem]{Definition}
 \newtheorem{example}[theorem]{Example}

\newcommand{\N}{\mathbb{Z}_{>0}}
\newcommand{\Z}{\mathbb{Z}}

\newcommand{\C}{\mathbb{C}}
\newcommand{\R}{\mathbb{R}}
\newcommand{\runter}[1] {\raisebox{-.45\height}{#1}}

\newcommand{\dif}{{\mathrm{d}}}

\newcommand{\uni}{\mathrm{U}}

\newcommand{\ee}{\mathrm{e}}

\DeclareMathOperator{\Wick}{\mathrm{Wick}}

\DeclareMathOperator{\Sym}{\mathrm{Sym}}

\DeclareMathOperator{\Tr}{Tr}

\newcommand{\bT}{\overline{T}}

\newcommand{\hp}[1]{^{(#1)}}

\usepackage{multicol,caption}

\tikzcdset{arrow style=tikz, diagrams={>={Stealth[round,length=4pt,width=4.95pt,inset=2.75pt]}}}

\usepackage{stackrel}
\newcommand{\gauss}{\dif \mu_0 (T,\bT)}

\newcommand{\conn}{^{\text{\tiny conn.}}}
\makeatletter

\newcommand*\notocchapter[1]{%
  \if@openright\cleardoublepage\else\clearpage\fi
  \thispagestyle{empty}\global\@topnum\z@
  \@afterindenttrue
  \let\@secnumber\@empty
  \@makeschapterhead{#1}\@afterheading
}

\newcommand{\CN}[1]{(\C^{N})^{\otimes #1}}
\newcommand{\orth}{\mathrm{O}}
\makeatother
\newcommand{\und}{\text{and }}
\renewcommand{\forall}{\text{for all }}
\newcommand{\with}{\text{with }}

 \title[Redundant moments of non-melonic random tensors \& bootstrap]{Redundant
 moments of non-melonic random tensor\\ models: simplifying the tensor bootstrap}
 \author[   C. I. P\'erez S\'anchez   ]{Carlos I. P\'erez S\'anchez }

  \address{ }
  \email{\href{mailto:perez.sanchez@protonmail.ch}{perez.sanchez@protonmail.ch}}

\begin{document}
\begin{abstract}
We prove the redundancy
of a family (the so-called melonic family) of large-$N$ moments in arbitrary tensor models. This
reduces the number of independent entries in
the positive semidefinite matrices that build the core of the tensor bootstrap
(for tensor models there are 
already several generalizations of the positive semidefinite Toeplitz and Hankel
matrices of lattice gauge theory and random matrix bootstrap, respectively). More concretely,
we prove that any melonic operator $C$ at large-$N$ has an expectation value that
depends on $C$ exclusively through its degree. A similar statement is shown here to hold
for the Schwinger-Dyson equations of melonic moments. The strength of our result is also its scope,
which is not limited to melonic tensor models.
\end{abstract}



\maketitle

\colorlet{rot_meloneB}{rot_melone}
\colorlet{rot_melone}{gruen_melone}


%
\section{Introduction}\label{sec:intro}

What we shall refer to as `random tensors' or `tensor models' \cite{GurauBook} appeared first in
\cite{AmbjornTM,SasakuraTM}, motivated by the success of random matrix theory as a description of
two-dimensional quantum gravity \cite{di19952d} (or in combinatorics, of the theory of
random maps \cite{EynardBook}). Random tensors were used to propose a theory
of simplicial gravity---or in combinatorics, of higher dimensional maps,
\cite{lionni2018colored,arXiv:2212.12200,octahedra}---whose large-$N$ expansion had to wait some decades \cite{Gurau:2009tw,Nexpansion_coloured}.
After the core of random tensors was fully developed,
other ensembles generalized the initial `$\uni(N)$ invariant tensor models'. In particular,
$\orth(N)$ tensor models \cite{ON} were also interesting for SYK-physics \cite{WittenSYK,GurauSYKinvN,KlebanovTarnopolsky,BenedettiGurauSYK}.
Random tensors served as well to develop renormalization techniques \cite{3Dbeta,4renorm,krajewskireiko,LahocheDine2018} or be approached to constructively
\cite{VignesRivasseau_makesense,RivDele}, just to name further research avenues they opened.
\\

In the present article, random tensors are considered from the viewpoint of
positivity bootstrap, which has been successful in lattice gauge field theory \cite{Kruczenski,KazakovZhengLattice}
 and random matrices \cite{Lin,KazakovZheng}. Positivity bootstraps were developed in \cite{BootstrapDirac} for
some 
\textit{bi-tracial} multi-matrix ensembles originated in random noncommutative geometry 
\cite{BarrettGlaser,bitracial,ReviewKhalkhali,RandomNCG_YMH,DArcangeloMC} and in
 \cite{perez2025loop} for gauge theories on quivers \cite{QuiversNCG}.
\\

It was natural to formulate the tensor bootstrap
as an open problem  \cite{IHP_OpenProblem},
exhibiting a particular positive semidefinite  matrix
that could play a role
in solving random tensor models, but it was only recently seen (luckily, not only by the author)
that indeed positive semidefiniteness---of this and of new matrices (cf. Sec. \ref{sec:BootsMatrices})---is able to solve a
restricted family of tensor models \cite{Reiko2026,TensorBootstrapProgram,Reiko_2BTM}.
Since the \textit{tensor bootstrap program} is still emerging,
we consider here some combinatorial aspects
that will simplify numerics, before bootstraps are applied to a completely
general tensor model.\\

Just as the moments of multi-matrix are indexed by  words
in the random matrices modulo cyclicity and other
symmetries inferred from the measure\footnote{We mean,
$AA,BB,ABAB,ABBA,\ldots$ for instance, if the symmetry is dihedral as in
the $ABAB$-model \cite{ABAB} or interpolations between $ABBA$ and $ABBA$ \cite{ABAB_MonteCarlo}.},
observables of
arbitrary $\uni(N)$-tensor models
are invariants built from polynomials in $T\in \CN D$, with each
factor $\C^N \subset \CN D$
being acted upon independently by $\uni (N)$ (we will omit the `$\uni(N)$' in `$\uni(N)$-invariance' from now on). After
reading the precise formulation below, it will be evident that
the complex conjugate $\bT$  of $T$ must appear in an invariant as often as $T$. \\
Let us illustrate with the
easiest case, $D=3$, how some invariants appear. After the product of three Kronecker deltas $\delta_{a_i,b_i}$,
the easiest invariant is $T\cdot \bT$ or
\[
T_{a_1,a_2,a_3}
\delta_{a_1,q_1}
\delta_{a_2,q_2}
\delta_{a_3,q_3}
\bT_{q_1,q_2,q_3}
\]
tacitly summed upon all $a_i,q_i=1,\ldots,N$ (Einstein sum). One particular way to construct invariants iteratively is
to replace each occurrence of
$\delta_{a_i,q_i} $ (for some $i=1,\ldots, D$) in an invariant
by \textit{dipoles}, that is matrices of the form
\[
T_{b_1,b_2,b_3}\Big[
\delta_{a_i,b_i}
\delta_{q_i,p_i}\textstyle\prod\limits_{k\neq i}
\delta_{b_k,p_k}\Big]
\bar T_{p_1,p_2,p_3}. \label{replacement_ofDelta}
                                       \]
At each step one introduces new deltas that can be replaced again by \eqref{replacement_ofDelta}, thus increasing
the degree by 2. Evidently, not all invariants are constructed like this,
for instance the next two are not: \vspace{-2ex}
\begin{subequations}%
\[ & \label{K33}
T_{i_1,i_{2},i_{3}}T_{n_1,n_{2},n_{3}}T_{x_1,x_{2},x_{3}}%
\bar{T}_{i_{1},x_{2},n_{3}} \bar{T}_{n_{1},i_{2},x_{3}} \bar{T}_{x_{1},n_{2},i_{3}}
\\
& \label{cube}
T_{i_1,i_{2},i_{3}}T_{n_1,n_{2},n_{3}}T_{x_1,x_{2},x_{3}}T_{z_1,z_{2},z_{3}}%
\bar{T}_{i_{1},n_{2},z_{3}} \bar{T}_{n_{1},i_{2},x_{3}} \bar{T}_{z_{1},x_{2},i_{3}} \bar{T}_{x_{1},z_{2},n_{3}}
\]\end{subequations}
By definition, those
constructed starting from $T\cdot\bT$ and replacing recursively,
for any $i=1,\ldots, D$, the Kronecker delta
$\delta_{a_i,q_i} $ by Term \eqref{replacement_ofDelta}
are called (connected) \textit{melonic}. Our result can thus be formulated as follows: \vspace{-3ex}
\begin{quote}%
\begin{theorem*}[Informally stated]\label{thm:informal} Large-$N$
moments of melonic observables
of the same degree (in $T$) coincide whenever they exist. Further, their
corresponding Schwinger-Dyson Equations coincide term by term.
\end{theorem*}\vspace{1ex}
\end{quote}
A first formal proof of this theorem  for
melonic tensor models (i.e. whose measure
is build only on melonic observables) is found in \cite{TwofoldUniversality}.
The present paper shows that the statement holds also for non-melonic models,
and for any $D$. Op. cit. shows a twofold universal behaviour of melonic tensor models, not only concerning
the independence  (whenever their degrees coincide) of the heavy combinatorics of
melonic interactions, but also universality in the number $D$ of tensor indices. According to  \cite{TwofoldUniversality}, once observables are melonic
and their moments convergent, then they become \textit{independent of $D$, too}, and can be computed
with a universal measure. Perturbatively, these are evaluated
via melonic polynomials (Sec. 3, op. cit.), to which any large-$N$ integral
of (products of) melons reduces. The first
universality type is proven to hold here for non-melonic models (but not the
$D$-independendence), too.\\

It is useful
to keep the information corresponding to the Kronecker deltas in each invariant,
and replace them colour-wise by strands. One also replaces the
tensors by nodes with attached half-vertices as in the following $D=3$-case:
\begin{subequations}\label{replacement_Inv2Graph}
\[
 \includegraphics[width=.6\textwidth]{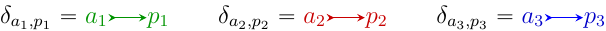}  \\
\runter{
\includegraphics[width=.5006\textwidth]{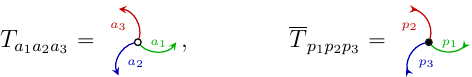}}
\]\end{subequations}%
With this rule, for instance, Invariant \eqref{cube} above is \vspace{-1ex}
\[ \label{cube2}
 \runter{\includegraphics[height=8ex]{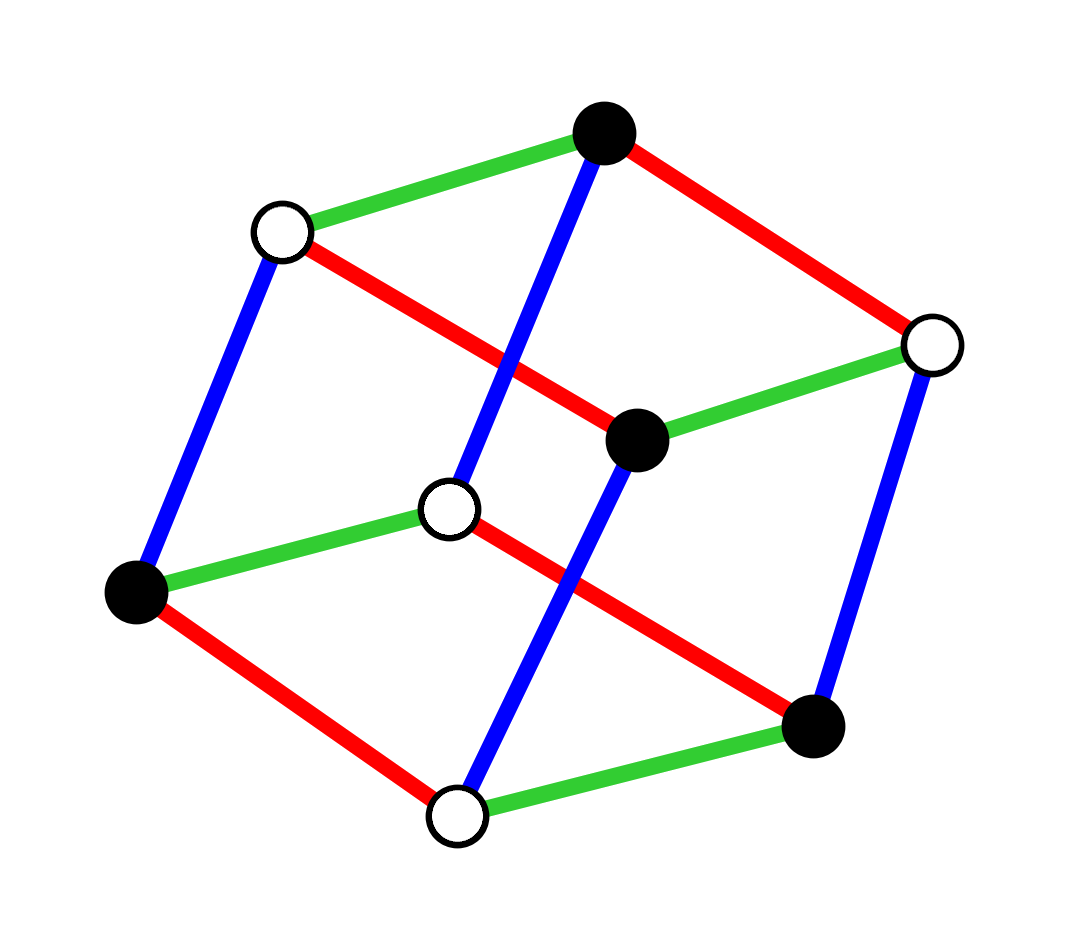}}
\]
An estimation of the magnitude of the simplification that our result implies
can be inferred from Table \ref{tab:Simplifyning}, in which we are allowed to collapse
all shown observables of the same degree (or number of vertices) to the single moment on the right.\\

To properly bootstrap a tensor model, the
positive semidefinite matrices $\mathcal M$ ($\mathcal M    \succeq 0$)
typically are neither of the
Hankel form (as for matrix models) nor Toeplitz form (of lattice Yang-Mills). It is important to be able to accommodate 
arbitrary observables in this matrix.
\\[.3ex]
\begin{minipage}{.58\textwidth}
 After proving our main results,
we present as outlook more general positive semidefinite matrices
and adapt those existing for non-melonic models.
For instance, if one wishes to obtain Observable \eqref{cube2} as an entry of $\mathcal M$, and simultaneously the
constant operator $1$ (needed in the future for relaxation \cite{KazakovZheng}),
we exhibit a basis of observables that leads to the positive semidefiniteness of
the matrix on the right.
\end{minipage}\hspace{-3ex}
\begin{minipage}{.418\textwidth}
\[ \notag\qquad
\E
\begin{bmatrix}
 N
  &     \raisebox{-.75ex}{\includegraphics[width=3ex]{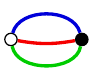} }
   &    \runter{\includegraphics[width=5ex]{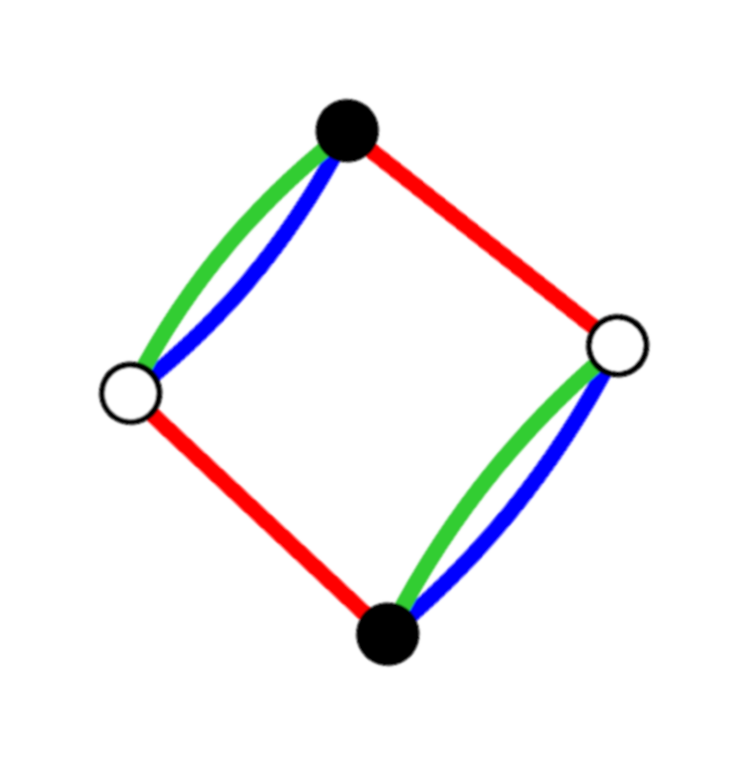} } \ldots
  \\
  \,\,\raisebox{-.75ex}{\includegraphics[width=3ex]{melon_tricolor} }
  &    \runter{\includegraphics[width=5ex]{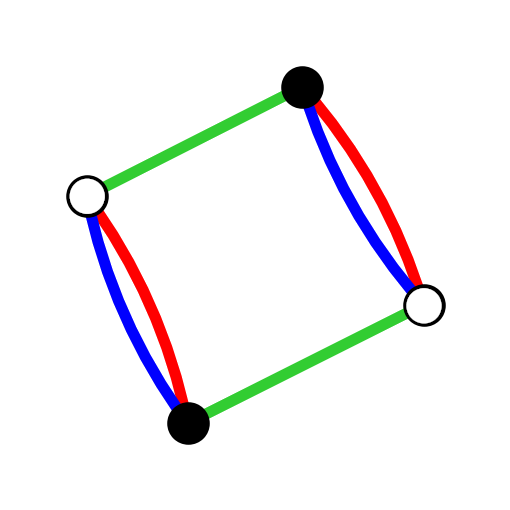} }
&  \runter{\includegraphics[width=5ex]{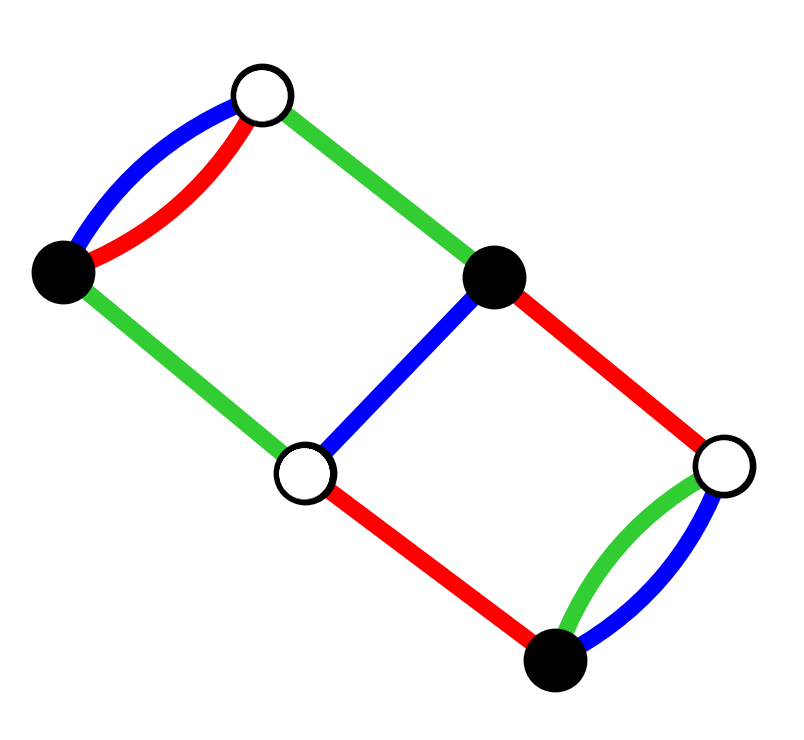} } \ldots
  \\
   \runter{\includegraphics[width=5ex]{V2colors} }
  &     \runter{\includegraphics[width=5ex]{E1colors} }
  &
 N \runter{\includegraphics[width=7ex]{Cubecolors} } \ldots \\
  \vdots &\vdots &\quad\vdots\qquad \ddots
\end{bmatrix} \succeq 0
\]
\end{minipage}

\begin{table}[b]
\[
\runter{
\includegraphics[width=.581\textwidth]{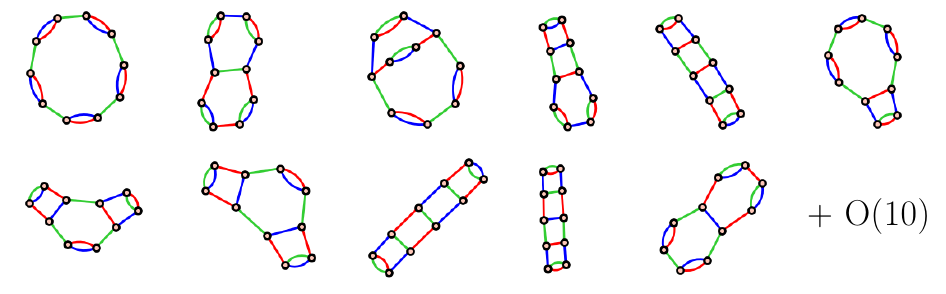}} \quad\vast\}  &
\qquad  m_{10} \nonumber\\
\runter{
\includegraphics[width=.8\textwidth]{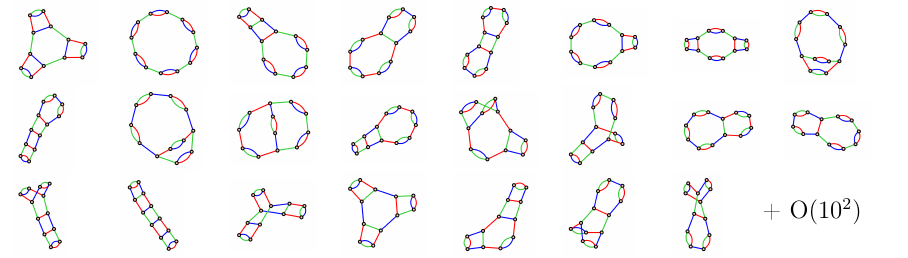}}\vastplus\}
& \qquad  m_{12} \nonumber
\]
\caption{Also for non-melonic tensor models, Theorem A simplifies all the grouped observables' large-$N$,
which can be reduced to a single 10-point $m_{10}$ or to a single 12-point moment (caveat: this does not mean that the moment
of \textit{all} degree-$k$ operators  are reduced to a single $m_k$, but only all melonic). All these graphs are evidently vertex bipartite, but
we have not coloured them black or white since
all these observables are real-valued. (To be continued in Table \ref{tab:nonBoteros}.)\label{tab:Simplifyning}}
\end{table}

% ...

\colorlet{rot_melone}{rot_meloneB}

\normalsize
\section{Tensor models}
\label{sec:setting_notation}

For integers $D\geq 3$ and $N\geq 2$ we consider the action of
$\uni (N)^D$ on
$\CN D$  by\begin{subequations}\label{trans_rule}\[
T_{a_1,a_2,\ldots ,a_D } & \mapsto  \prod_{k=1}^D
U\hp k_{a_i,b_i}
T_{b_1,b_2,\ldots b_D } \\
\bT_{q_1,q_2,\ldots q_D }& \mapsto  \prod_{k=1}^D
\overline{U}\hp k_{q_i,p_i}
\bT_{p_1,p_2,\ldots, p_D }
\]\end{subequations}
for all $
(U\hp 1,\ldots U\hp D) \in \uni(N)^D$ and $T=\CN D$,
whose components are $T _{a_1,a_2,\ldots a_D } $, where sum of repeated indices is implicit. (The very transformation rule
makes it irrelevant which bases we pick, as indices appear always contracted.)
Rule \eqref{trans_rule} is originally from \cite{Gurau:2009tw}.

\subsection{Tensor measures}
Measures in random tensors are build from a finite set $B_1(T,\bT),$ $B_2(T,\bT),\ldots, B_m(T,\bT)$
of unitary invariants and read $\ee^{-S\inter(T,\bT)} \gauss$ with
\[\ee^{-S\inter(T,\bT)} \gauss \qquad \qquad
\with
\qquad S\inter(T,\bT) &= \summ j  g_m N^{s_j} B_j(T,\bT). \label{Model}
\]
We now explain each ingredient. First, $\gauss$ is the normalized ($\int \gauss =1$)
Gaussian measure on $\CN D$ with
\[ \int_{\CN D}
T_{a_1,a_2,\ldots ,a_D } \bT_{p_1,p_2,\ldots, p_D } \gauss  = \frac{1}{N^{D-1}}\prod _{c=1}^D
\delta_{a_c,p_c} . \label{correlation_gaussian}
\]
Above $s_1,\ldots,s_m \in \Z_{\geq 0}$
are scaling factors that ensure the existence of the large-$N$ limit and $g_1,\ldots, g_m \in\R$
are coupling constants or just real parameters of the measure.
Since the Gaussian measure (which includes the `kinetic' quadratic
part $T\cdot\bT$) is fixed we refer to $S\inter$ as the tensor model.
A model is called melonic if $S\inter$ is spanned
exclusively by (connected) melonic invariants $B_1,\ldots,B_m$. For a
general theory the parameters $s_1,\ldots,s_m$ are, to the best of the authors' knowledge, not known.
For melonic models, it is known that these scalings are all $D-1$ \cite{GurauRyan}.
(Below it will be clear that the role of  $s_1,\ldots,s_m$ is for our results unimportant, as we focus
on maximizing faces-graphs.) As the $B_j$'s, the $s_j$'s and $N$ will be fixed in the whole manuscript,
we use a light notation for $S\inter$, instead of writing $S\inter_{N;g_1,\ldots,g_m}$.

The moment $\E[B]$ of an observable $B(T,\bT)$, shorthand for
$\E_{g_1,\ldots,g_m}\hp N [B(T,\bT)]$,  is
\[ \label{ExpectDefinition}
\E [B ] =\frac{1}{\mathcal Z}\int_{\CN D}
B(T,\bT) \gauss,
\qquad \mathcal Z = \int_{\CN D }\ee^{-S\inter(T,\bT)}  \gauss,
\]
but this exact expression need not be finite at large-$N$.
Melonic observables require to be scaled by $1/N$
for their moments to be finite. (For non-melonic models
it is not know how to find the optimal way to scale
the moments for them to be finite, but for melonic models
\cite{TensorBootstrapProgram} presents an algorithm.)
Assuming that $\E[B]$'s leading order scales as $N^{q_B}$ at large-$N$, its
large-$N$ moment $ \lim_{N\to \infty } N^{-q_B}\E[B]$ is denoted by $m (B)$.
(While we are introducing this notation
for moments, we anticipate that, when $B$ is melonic,
we will be able, as a consequence of the present results,
to replace $m (B)$ by $m_{2p}$, with $2p=\deg B$.)

\subsection{Observables as graphs}
It has been anticipated in the Introduction that
graphs absorb the combinatorics of tensor invariants.
In fact, the map that applies Rule \eqref{replacement_Inv2Graph} to an invariant
yields a bijection between invariants and $3$-coloured graphs
(with the obvious generalization for any $D$). Concretely, the target consists of
regular graphs with vertices that are either \textit{black} (or, when enumerated, \textit{even})
and \textit{white} (or \textit{odd}), and edges that are
regularly coloured
by the set $\{1,2,\ldots, D\}$ and connect only black with white vertices (`$D$\textit{-coloured graphs}').
It is evident that there is a \cite{GurauRyan}
one-to-one correspondence between
invariants of tensors with $D$-indices, in the sense of \eqref{trans_rule}, and  $D$-coloured graphs. \\

As usual, $V(B)$ is the
\textit{set of vertices} of $B$,
and here we write $V_0(B) \subset V(B)$ for the
black vertex set and $V_1(B)\subset V(B)$ for the white vertex set.
To avoid confusion with the edge-colouring we
use \textit{parity} for the belongingness to either set.
To distinguish  a coloured graph
from an invariant we use $B$ and $B(T,\bT)$, respectively, if really needed.
If the $B$ is connected, $B(T,\bT)$ is \textit{single-trace} or itself  \textit{connected}.
\\

Since all invariants $B$ are bipartite, a useful and wide-spread
notation we adhere to is $p(B)=\frac12 \#V(B)$.
A \textit{vertex-labelling} $\lambda=(\lambda_0,\lambda_1)$ of a coloured graph $B$
is a pair of bijections
\[ \lambda_0 :V_0(B) \stackrel\sim\to \{0,2,\ldots, 2p-2\} \qquad \und\qquad
\lambda_1:  V_1(B) \stackrel\sim\to\{1,3,\ldots, 2p-1\}.
\]


\subsection{Wick contractions and Feynman graphs}
This short section introduces the notation to evaluate the integral
\[ \label{WickIntegral}
\int_{\CN D}
B_{1}(T,\bT)\cdots B_n(T,\bT) \gauss =
\sum_{\pi \in \Wick(\mathcal B) } A [\pi (\mathcal B )]
\]
by Wick's theorem. Here the sum is performed over \textit{Wick contractions}
of the disconnected graph $\mathcal B =\dot\cup_{i=1}^n B_n$,
namely bijections $\pi: V_1(\mathcal B ) \to V_0(\mathcal B )$.
In the RHS one sums over amplitudes $A(G_\pi)$ of \textit{Feynman graphs}
$G_\pi= \pi(\mathcal B)$, where $ \pi(\mathcal B)$ denotes the
$(D+1)$-coloured graph determined as follows ($A$ is defined below). Let $E(\mathcal B)$ be the
edge set of $\mathcal B$, then
\[
V(\pi(\mathcal B)) & := V(\mathcal B) \\
E(\pi( \mathcal B ) ) & :=E  (\mathcal B) {\dot\cup} \{ (v,\pi(v) ) : v\in V_1(\mathcal B) \}
\]
with the set of new edges (the `graph of the function' $\pi$)
is given, by definition the \textit{colour $0$}. Conventionally,
we colour it with gray (other conventions use a dashed line). The 0-colored edges
$\{ (v,\pi(v) ) : v\in V_1(\mathcal B) \}$ are also called \textit{propagators}.
\\

A connected subgraph of $\pi(\mathcal B)$ that contains only edges of colours $\{0,c\}$
is called a $c$-\textit{coloured face} of $\pi(\mathcal B)$. (Caveat: several sources
allow faces to have colours  $\{c_1,c_2\}$ for $c_1\neq c_2$, $c_i\in \{1,\ldots,D\}$,
but we will not use this convention). The total number of faces (i.e. adding those for colours $c=1,\ldots,D$)
will be abbreviated by
\[  f(G_\pi) :=  \# F (G_\pi) .\]

We are finally in the position of defining the amplitude $A(G_\pi)= A [\pi (\mathcal B )]$
as 
\[A(G_\pi) = \label{amp_def}
N^{f(G_\pi) -(D-1)p(G_\pi)} \qquad \text{[recall $p(G_\pi)=\tfrac12 \#V(G_\pi)$]}.
\]
Just as in \cite{TwofoldUniversality}, we stress that
we perform Sum \eqref{WickIntegral} over Wick contractions and not over graphs,
which should explain why we have a coefficient $1$ in the amplitude. This expression explains also the
traditional notation  `$p$' in Eq. \eqref{amp_def} as the number of propagators, $p(G_\pi)$.
The reason to be interested in faces is the equivalence of
Eq. \eqref{correlation_gaussian} with $D$ parallel lines, one per colour. One can then decompose
the 0-colour (gray) concretely as
\begin{equation}\notag
\frac{1}{N^{D-1}}
\raisebox{-.35\height}{
\tikz[scale=.65,  outer ysep=-2 pt,
    baseline=-.5ex,shorten >=.1pt,node distance=18mm,
    semithick,auto,
    every state/.style={fill=black ,draw=black,inner sep=.4mm,text=black,minimum size=0},
    accepting/.style={fill=gray!50!black,text=white},
    initial/.style={white,text=black}
]{
\node[state,minimum size=0,circle, fill=white,  draw=green!58!black] (A11) at (0,0) {};
 \node[state,minimum size=0,circle, fill=white, draw=red] (A12) at (0,.4) {};
 \node[state,minimum size=0,circle, fill=white, draw=blue] (A13) at (0,.8) {};
  \node[state,minimum size=0,circle, fill=white, draw=orange] (A1D) at (0,1.6) {};
  \node[ ] (mid) at
(1.5,1.32) {\tiny $\vdots$};
\node[state,minimum size=0,circle , draw,fill=blue!40] (q13) at
(3,.8) {};
\node[state,minimum size=0,circle , draw,fill=orange] (q1D) at
(3,1.6) {};
 \node[state,minimum size=0,circle, draw,fill=red!50] (q12) at (3,.4) {};
 \node[state,minimum size=0,circle,draw,fill=green!40!lightgray] (q11) at (3,0) {};
 \foreach \c/\j in {green!58!black/1, red/2, blue/3,orange/D}{
 \draw[,\c] (A1\j) -- (q1\j);
 }
 } }\,
 \label{prop}
\end{equation}
and then gather all $c$-coloured loops $c=1,\ldots,D$ of these delta's with those of the integrand.

\begin{minipage}{.69\textwidth}\begin{example}\label{K33K33}
In the  Feynman graph (top right) an example of what we call face of colour 2 (red) is the four red-gray alternating edges
at $(v,w,y,x)$.
Having in total 9 faces, three per colour,
and $(D-1)p=12$ in this case, the amplitude of this graph is
$N^{-3}$.  Nevertheless, this is a dominant graph
in the model with a single interaction vertex $K_{3,3} $,
for the interactions in tensor integrals appear actually weighed
with powers of $N$.
Here, the
graph $K_{3,3}$ has the scaling $N^3$, or $s(K_{3,3}) =3$,
and the moment $\E [ K_{3,3} ]$ scales as 1, or $q(K_{3,3})=0$.
(See e.g. \cite{LionniJohannes} to know more about this model.)
\end{example}
\end{minipage}
~
\begin{minipage}{.28\textwidth}
 \[\runter{\includegraphics[width=19ex]{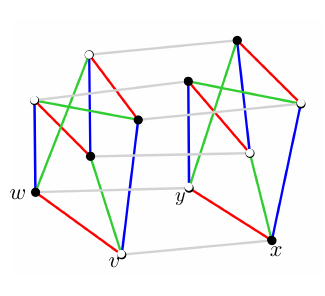}} \notag \\
 K_{3,3} = \runter{\includegraphics[width=8ex]{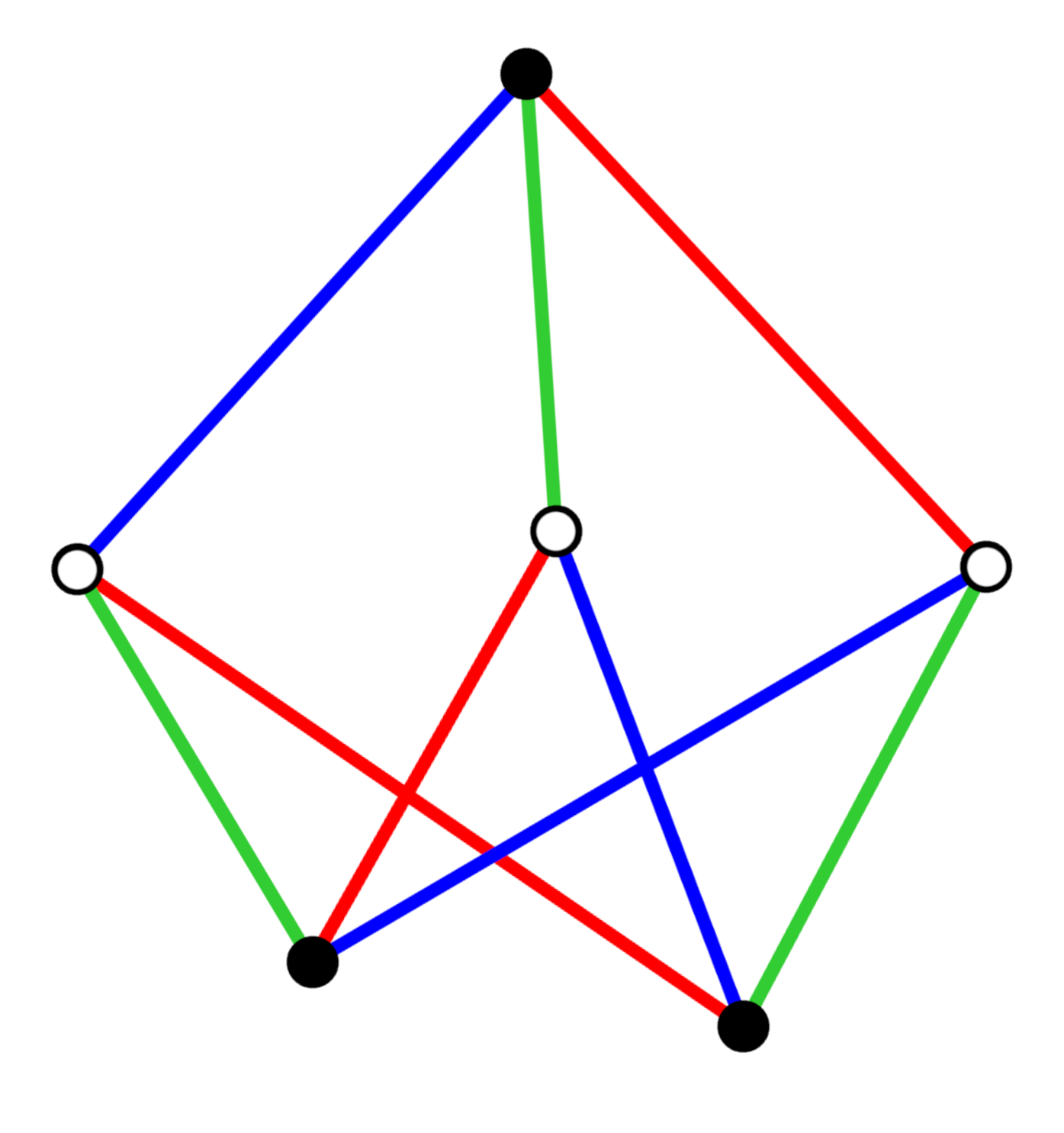}} \notag
 \]
\end{minipage}

\section{A canonical tree for melonic graphs}

We need to elaborate on the well-known fact that `melons are trees',
exploited for instance in \cite{BranchedPolymers}. Concretely, a single melon is actually several trees.
Lacking uniqueness, we need to
find a way to select a canonical tree, together with a canonical vertex-numeration, that allows us to
`replace a melon by another melon' with the same number of vertices in a unique way.
First, we do not order all melon labels, but define a criterion
of admissibility. Among all admissible labels of a melon, we define
a canonical one.
It will be enough for us that there is a unique map that brings
any label, even if not admissible, to the canonical label.

\subsection{Dipole insertion vertex pairs (\dip s)}
\label{sec:melonic_univ}
A connected melon is either the quadratic invariant
$T_{a_1,a_2,\ldots, a_D} \delta_{a_1,p_1}
\cdots \delta_{a_D,p_D} \bT_{p_1,p_2,\ldots, p_D}$ or iteratively
obtained from that invariant by replacing, at each step,
a delta $\delta_{a_i,p_i}$  by a \textit{dipole of colour $i$}, namely by \[
T_{b_1,b_2,\ldots, b_D}\Big[
\delta_{a_i,b_i}
\delta_{q_i,p_i}\textstyle\prod\limits_{k\neq i}
\delta_{b_k,p_k}\Big]
\bar T_{p_1,p_2,\ldots,p_D}. \]
Equivalently, but more colourfully, a $D$-coloured graph  $B$ is a \textit{melon} if it
is either the quadratic invariant, i.e. $B= \runter{\includegraphics[height=5.75ex]{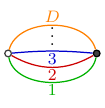}}$, or
it is the result of substituting there any of the coloured edges
\[ \hspace{-3ex}
\includegraphics[height=3ex]{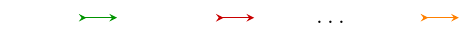} \label{lines}
\]
iteratively by the \textit{dipole} of its respective colour,
\[
\raisebox{-2.5ex}{\includegraphics[width=14ex]{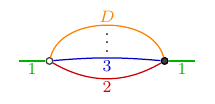}}\runter{
\includegraphics[width=14ex]{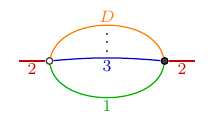}}
\cdots
\runter{\includegraphics[width=14ex]{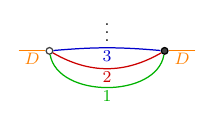}} \label{dipoles}
\]

We call each of the vertex pair that emerged by a substitution of
any line \eqref{lines} by its corresponding dipole \eqref{dipoles}  a
\textit{dipole insertion vertex pair} (\dip).  Observe that
one can obtain the complete (not ordered) list of \dip s
of a melon $C$
as follows: at the beginning, the list of \dip s consist
of all actual dipoles of $C$ (by definition, always exist at least one dipole). Then,
one removes one by one, after which new dipoles should be formed
after having glued the broken edges caused by dipole removal, that is, when going from \eqref{dipoles} to \eqref{lines}.
Observe that dipoles always exist, else the new
graph was not a melon, and then $C$ was not a melon either. One adds these new dipoles to the list and
repeats this process until one gets a 2-vertex graph, whose
two vertices are the last \dip.
\\

\begin{minipage}{.67\textwidth}\begin{example}
In the graph on the right, with $D=6$, the six  actual dipoles are obviously \dip s,   but, more interestingly,
also the couple of vertices in the middle tagged with $(v,w)$,
which do not even have edges in common, is a \dip,  since
removing the dipoles and contracting the broken half-edges
 yields the two-vertex graph with six colours spanned by $(v,w)$.\\ 
\end{example}
\end{minipage}
~~
\begin{minipage}{.30\textwidth}
\[ \notag
\raisebox{5ex}{
\includegraphics[height=20ex]{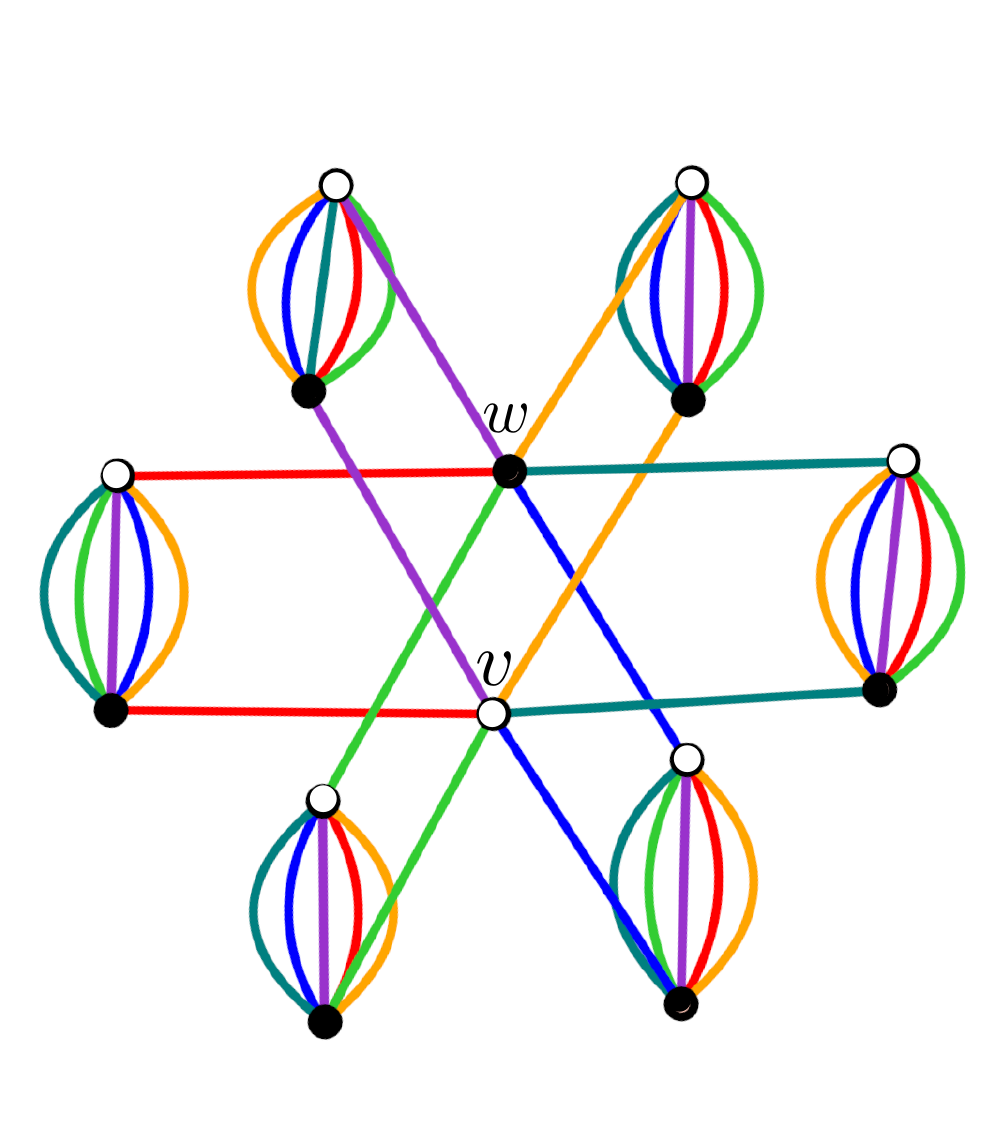}}
\]
\end{minipage}

\allowdisplaybreaks[3]

\subsection{Ordering labels of melons}\label{sec:new_recipy}

The \textit{tree $\mathcal T_{\lambda}(C)$ associated to a melon} $C$
with vertex labels $\lambda$ is obtained as follows.
Since $C$ has  $2p=2p(C)$ vertices, and $C$ is a melon,
we can pick any of the available dipoles, and let us call
$\lambda (v_{p-1})$ the number of the corresponding black vertex $v_{p-1}$ (we mean $\lambda_0$,
which enumerates vertices using even numbers but keep notation compact). Let us denote that dipole's colour by
$c_{p-1} \in \{1,\ldots,D\}$. After the picked dipole
is replaced by a single line of colour $c_{p-1}$, as
\eqref{dipoles} $\to $ \eqref{lines}, one records the
pair $(\lambda(v_{p-1}),c_{p-1})$ and picks another dipole (if it did
not exist, it surely will: after collapsing the dipole to a line
a dipole must emerge by melonicity of $C$). This can be iterated
until one arrives to the two-point graph associated to $T\cdot \bT$.
In the process, we recorded  the pairs
\[
( \lambda(v_{p-1}),c_{p-1}),
(\lambda(v_{p-2}),c_{p-2}),\ldots,
(\lambda(v_{1}),c_{1}) \in \{0,2,\ldots,2p-2\} \times \{1,2,\ldots,D\}
\]
which allow to construct a vertex-labelled $D$-ary tree in
by attaching to the root $ v_0 $
the vertex $v_1$ at its only edge of colour $c_1$,
and so on,
attaching the vertex $v_i$ at its only edge of colour $c_i$
to $v_{i-1}$.  This process is depicted by the transition of Figure \ref{fig:drzewo0} (a) to Figure \ref{fig:drzewo0} (b).
\begin{definition}
A branch of a rooted tree is any path from the root (placed at the top)
to a last-parent, i.e. a node
having only children without children.
A vertex-labelling $\lambda$ of a melon is \textit{admissible}
if its tree $\mathcal T_{\lambda}(C)$ (vertex-labelled
by even numbers  as induced by $\lambda$)
is  increasing along any branch $\theta$, and if the numeration is either consecutive (in the even numbers)
when moving downwards, or else the next number leaves a branch, but
none of the successors enumerates another node of original branch $\theta$.
\end{definition}
Intuitively, admissibility says that if we follow the path of the
numeration, the path goes either downwards (if we draw the root on the top)
or it is allowed to go side- and upwards, but only in order to definitely change branch; see Figure \ref{fig:branch_and_admissible_not}.
\begin{figure}[H]
 \includegraphics[width=48ex]{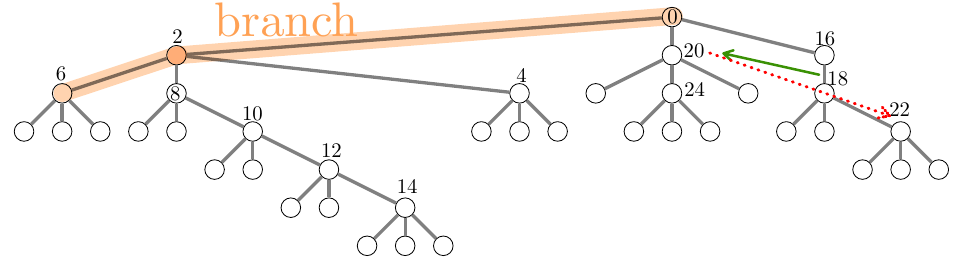}\caption{We show a 3-ary rooted tree, a branch (0,2,6)
 and a non-admissible numeration. `Last-parents' are $4,6,14,22$ and $24$. All labels up to (incl.) $20$ are placed legally, as it is possible to leave a branch,
 as in the transition $18\to 20$, but never allowed to come back (forbidden is in red/dashed).
 \label{fig:branch_and_admissible_not}}
\end{figure}\vspace{-2ex}

Consider the set $\Lambda_D = \{ (l,\alpha) : l =0,1,2\ldots \und  \alpha = 0,1,\ldots,D^p-1 \}$. Since $\Lambda_D$ will index all possible $D$-ary trees,
we like to think of this set as a subset of a lattice on the plane in polar coordinates [in gray in Fig. \ref{fig:drzewo0} (c)] .
For a fixed melon $C$, we associate to a labelling of $V(C)$
---or equivalently, to  $\mathcal T$, the tree that originates $C$---
an arborescence (for us: a directed edge $D$-colored tree) that embeds in $\Lambda_D$
in a natural fashion, as exemplified for the tree $\mathcal T \subset \Lambda_3$
in Figure \ref{fig:drzewo0}. To form the arborescence  $\mathcal A$
one removes all leaves from the tree $\mathcal T$, and keeps only those edges
that became parents. The way we embed the tree $\mathcal T$ in $\Lambda_D$
respects the spaces of the leaves not present in $\mathcal T$. To be concrete, first, the coordinate $l$
represents the depth of each branch of $\mathcal T$  (depth: number of
children, branch: any root to end-point path). Let us construct the arborescence $\mathcal A$
inductively, supposing that
$(l,\alpha)$ is already the target of a directed edge in $\mathcal A$.
Then, by definition, there is an edge from  $(l,\alpha)$ to
$(l+1,\beta)$ whenever
\begin{itemize}
 \itemb
$\beta= D \alpha + r$,
for some $r= 0,\ldots,D-1$,  that is
when $\lfloor \beta /D \rfloor =\alpha$, \textit{and}
\itemb
the vertex in $\mathcal T$ that corresponds to  $(l,\alpha)$
has a colour-$(r+1)$ descendent.
\end{itemize}  We assumed that   $(l,\alpha)$
is target of a directed edge in $\mathcal A$, so
we know  that $(l,\alpha)$ corresponds to a parent $v$ in $\mathcal T$, thus
the second  condition says
that $v$ must be parent of a colour-$c$ child with $c=r+1$.
The remainder $(c-1)$ of the division
determines the color $c$ of the edge (to match the conventions of reserving
the colour $0$ for propagators, this
colour cannot match the remainder exactly and---somehow annoyingly---has to be shifted by $1$).\\

\begin{figure}
\begin{minipage}{.20\textwidth}
\includegraphics[width=.79\textwidth]{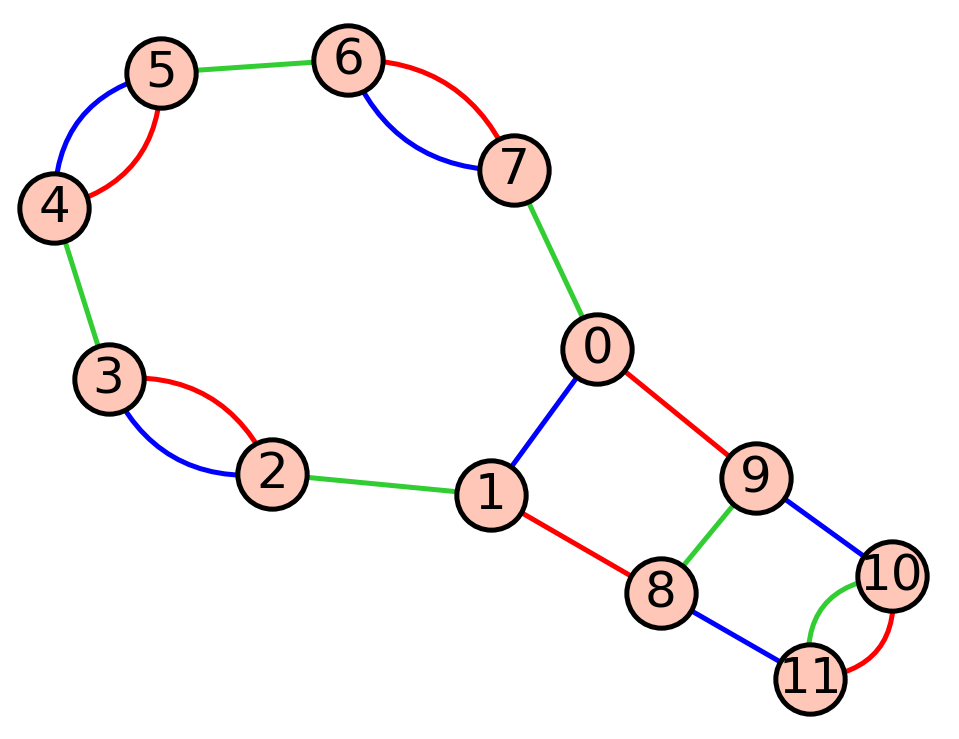}
\end{minipage}~
\begin{minipage}{.3\textwidth}
\includegraphics[width=.95\textwidth]{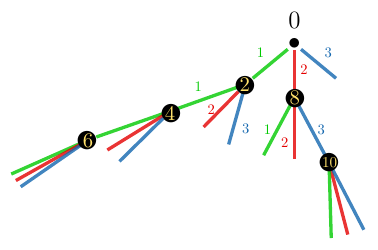}
\end{minipage}~
\begin{minipage}{.55\textwidth}\qquad
\includegraphics[width=.7\textwidth]{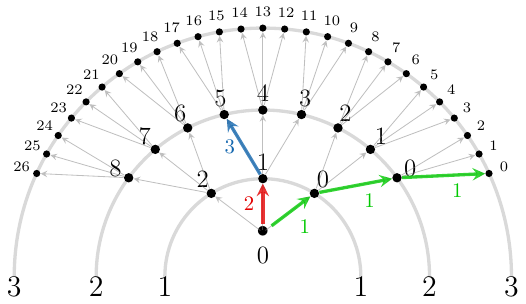}
\end{minipage} \small
\phantom{A}\hspace{0ex} (a) $C$'s label $\lambda$
\phantom{A}\hspace{9ex} (b) Tree $\mathcal T_\lambda(C)$
\phantom{A}\hspace{4ex} (c) Arborescence $\hookrightarrow \Lambda_3$ (black dots $= \Lambda_3$)
\caption{A melon $C$ with vertex label $\lambda$ (a)
can be constructed from a tree (b), from
which also an arborescence
is constructed and depicted as embedded in $\Lambda_3$ in (c). When (b) is turned upside down and
the leaves of the tree pruned, we obtain (c).
The first coordinate of $(l,\alpha)\in \Lambda_3$ are denoted by black numbers $l=0,1,2,3$.
The second coordinate $\alpha$ is the dot's number right to left
at depth $l$. The label (a) of the melon yields here two endpoints
with coordinates (2,5) and (3,0).
\label{fig:drzewo0} }\normalsize
\end{figure}



\begin{minipage}{.74\textwidth}
\begin{example}\label{ex:drzewa}Consider the melon $C$ on the right.
Below, three admissible labellings of $V(C)$ we display the associated arborescences.
In each picture, $r_m$ denotes the remainder of the $m$-th division by $D=3$
(the depicted colour corresponds with $r_m+1$). For instance
the first example below has two end points. The longest branch
yields $16= \sum_a r_a 3^a = 1\times 3^2 +2 \times 3^1 + 1 \times 3^0$, and the other just $0=r_0' \times 3^0$.
\end{example}
\end{minipage}~
\begin{minipage}{.20\textwidth}
\[ \notag
\phantom{ABC}
C = \hspace{-3ex}
\runter{
\includegraphics[width=.8\textwidth]{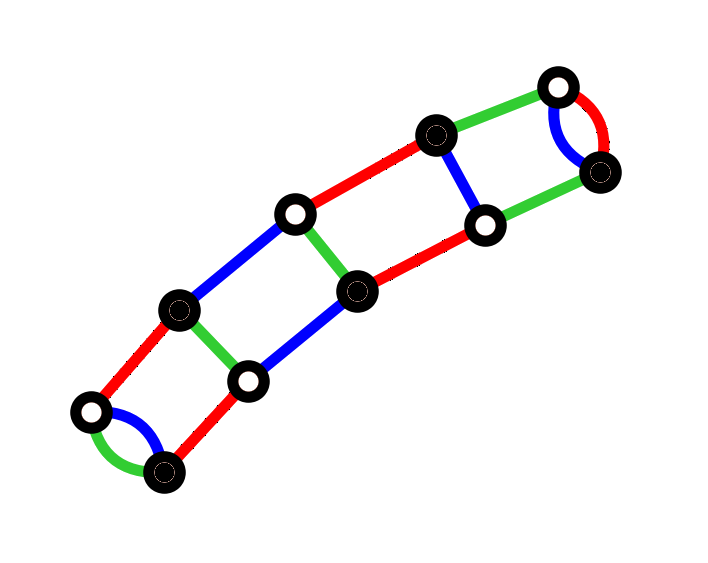}}
\]
\end{minipage}

\begin{subequations}
\[
\runter{
\includegraphics[width=.25\textwidth]{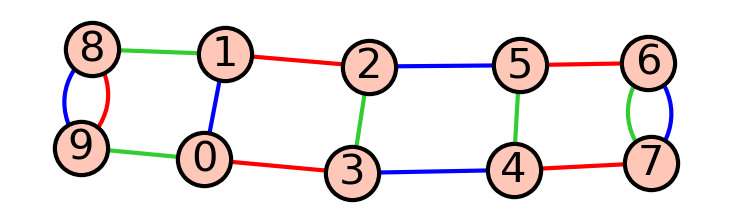}} & \quad \to \quad & \hspace{0ex}
\runter{\includegraphics[width=.48\textwidth]{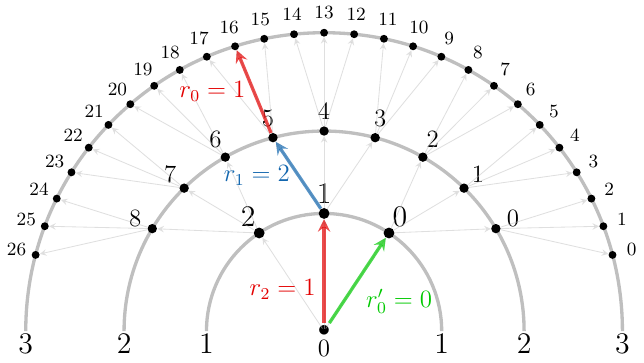}}
\label{ex_branches0}
\]
 Similarly for the next two
examples, each with a single branch.  \[
\runter{
\includegraphics[width=.25\textwidth]{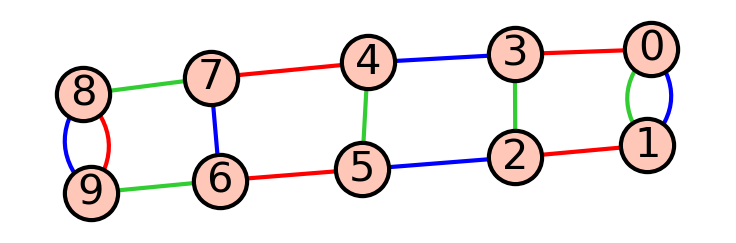}} & \quad \to \quad &
\runter{\includegraphics[width=.5\textwidth]{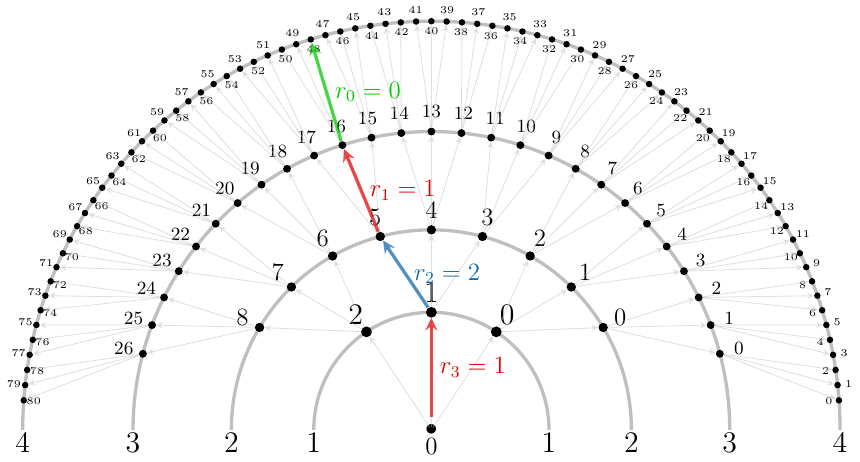}}
\label{ex_branches1}\\
\runter{
\includegraphics[width=.25\textwidth]{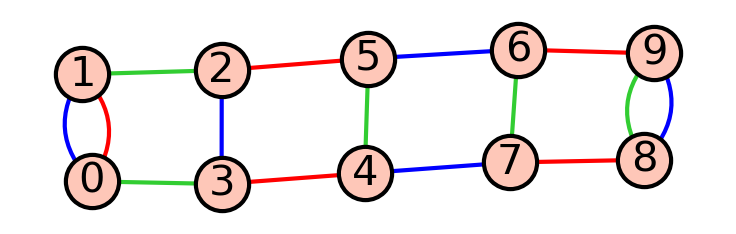}} & \quad \to \quad &
\runter{\includegraphics[width=.5\textwidth]{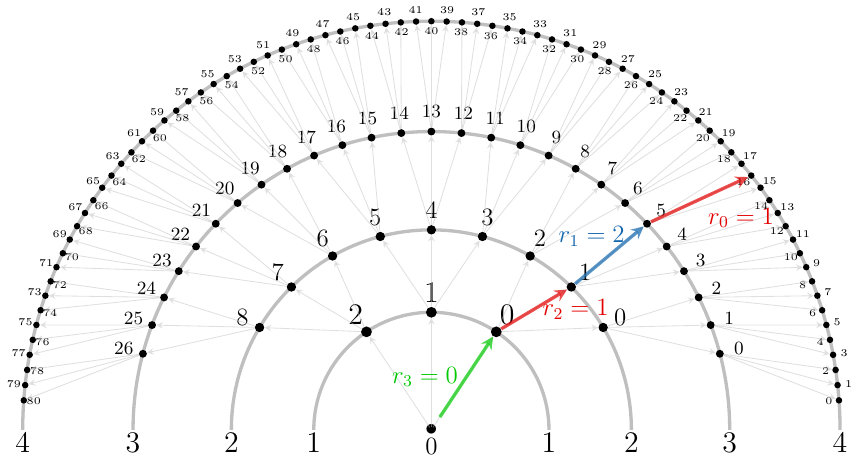}}
\label{ex_branches2}
\]%
\end{subequations}%
Notice that we care about admissible labels. An example of non-admissible label is
\[\runter{
\includegraphics[width=.205\textwidth]{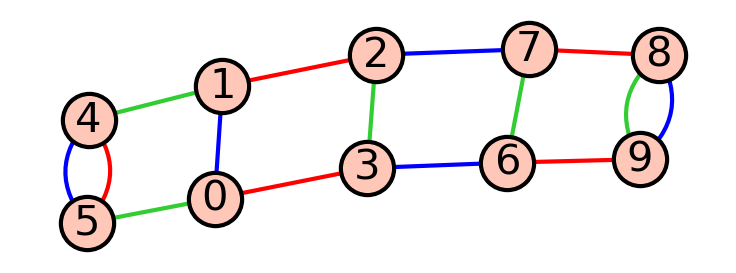}}\]
since starting with the root (0),
the labels of the right branch `jump' from $\{2,3\}$ to $\{6,7\}$, while the missing
$\{4,5\}$ is on the left. In the tree associated to this melon, this is a jump to another branch, and trying to come back to the original. This example also motivates
the following definition.
\begin{definition}\label{def:order_melonic} First, let us consider
$\Lambda_D$ with its lexicographic order
\[
(l,\alpha) > (l',\alpha') \qquad  \wenn  l> l' \text{ or if } l=l' \,\,\und  \alpha>\alpha'.
\]
We can extend this order to  $\Lambda_D^k$, $k\in\N$.
Consider two vertex-labels $ \lambda$ and $\lambda'  $
of the same $D$-coloured melon $B$, and let
\[
(l_1,\alpha_1),
(l_2,\alpha_2),&
\ldots,
(l_b,\alpha_b) \\
(l_1',\alpha_1'),
(l_2',\alpha_2'),&
\ldots,
(l_b',\alpha_{b'}')
\]be
the coordinates in $\Lambda_D$
of the $b$ (resp. $b'$) endpoints of $\lambda$ and $\lambda'$ respectively, also lexicographically ordered
\[ (l_1,\alpha_1)> (l_2,\alpha_2)
> \ldots > (l_b,\alpha_b)  \,\,\und
 (l_1',\alpha_1')> (l_2',\alpha_2')
> \ldots > (l_b',\alpha_b') . \]
Then, by definition, $\lambda>\lambda'$ \begin{itemize}
                                            \itemb
if there exist an $i=1,2,\ldots,\min\{b,b'\}$ with
$(l_i,\alpha_i) > (l_i',\alpha_i')$
\itemb \textit{and} if $i$
is the minimal index for which equality
 does not hold, i.e.
$ (l_a,\alpha_a) = (l_a',\alpha_a')$ for all $a<i$.

                                           \end{itemize}
Having this order at our disposal, we can define \textit{the
canonical label} of a melon $C$ as
$\lambda^{\can}(C)= (\lambda_0^{\max},\lambda_1^{\max})$, where
$\lambda_0^{\max}$ is maximal for the black (or even) vertices in the above sense, induced by
the lexicographic order of its coordinates in $\Lambda_D$, and (see Sec. \ref{sec:setting_notation}
for notation)
\[\lambda_1^{\max}(v) := \lambda_0^{\max}(w) + 1,\,
 \text{ whenever $(v,w)$ is a \dip{ }of $C$}. \label{lambda1}
\]

\end{definition}

Before verifying that this definition makes sense, let us digest it by retaking the example given above.

\begin{example} In Example \ref{ex:drzewa} we have given three
arborescences (corresponding to three trees) of the
melon denoted by $C$ there.
\begin{enumerate}[(a)]
 \item The vertex labelling in \eqref{ex_branches0}
 has two endpoints, whose
 coordinates in $\Lambda_3$ are
 $(l_1,\alpha_1) = (3,16) $ and $(l_2,\alpha_2)=(1,0)$.
 This corresponds to the ternary
 $(121)_3=16$ and $0_3=0$, where the RHSes are in the decimal
 representation.
 \item In \eqref{ex_branches1} one has a single endpoint
 with coordinate $(4,48)$ corresponding to $48=(1210)_3$.
  \item In \eqref{ex_branches2} has coordinates
  $(l,\alpha)= (4,16)$, with $16 = (0121)_3$.
\end{enumerate}
Notice that the last `sixteen'
  has $l=4$  digits, as $l$ is determined by the arborescence,
  whence we have a digit 0 on the left; this 0 is relevant, as
  in the lexicographic order $(4,16) > (3,16)$, the latter being the sixteen of \eqref{ex_branches0}.
One has  then the following  order for these vertex labellings,
$\lambda|_{\text{(b)}} > \lambda|_{\text{(c)}}> \lambda|_{\text{(a)}}$. One has
$\lambda^{\max} (C) = \lambda|_{\text{(b)}}$, as this is the largest number with four
digits in $\{0,1,2\}$ (corresponding to colours $\{1,2,3\}$) that yields the original melon $C$.
Observe also that, although $ (2222)_3 >  (2110)_3 >  (1210)_3$,
the arborescences with coordinates $(2222)_3 $ and $(2110)_3$ belong to different melons, namely to
\[\runter{
\includegraphics[width=.12\textwidth]{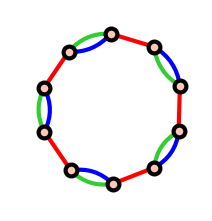}} \to (2222)_3  \qquad \quad\runter{
\includegraphics[width=.133\textwidth]{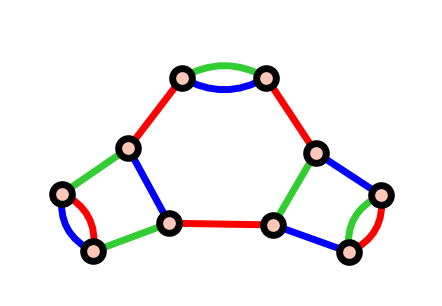}}\to (2110)_3
 \]
This emphasizes that we can compare labels only for a fixed melon,
and in the case at hand the maximal label is \eqref{ex_branches1}.


\end{example}

\begin{claim}\label{thm:unicity_max_claim}
The canonical label of a melon $\lambda^{\can}(C)$ is, indeed, unique.
\end{claim}
\begin{proof}
We defined $\lambda^{\can}(C) = (\lambda_0^{\max},\lambda_1^{\max})$
where $\lambda_0^{\max}$ is determined by its coordinates in $\Lambda_D$ being maximal,
and since a melon knows its \dip s, $\lambda_0^{\max}$ determines $\lambda_1^{\max}$, cf. Eq. \eqref{lambda1}.
But by construction of $\lambda_0^{\max}$, the $\Lambda_D$ coordinates of
$\lambda_0^{\max}$ are a list of $b$ base-$D$ integers (set $p=p(C)$ here)
\[ \label{Dary_proof}
(c_{0,\mu} c_{1,\mu} c_{2,\mu}\cdots  c_{l_\mu-1,\mu} )_D
=\sum_{a=0}^{l_\mu -1} (c_{a,\mu} -1) D^a \qquad \mu=1,\ldots,b
\]
where $b$ is the number of endpoints of $\lambda_0^{\max}$ in $\Lambda_D$, and
$c_{i,\mu}$ is the colour of the $i$-th dipole insertion of the $\mu$-th branch
that goes from the root (0) to the $\mu$-th endpoint (recall $c_{a,\mu} -1$ are the  remainders of division by $D$).
Among all these numbers there is exactly one with the  largest $l_\mu$ and largest value of the RHS of Eq. \eqref{Dary_proof}.
Uniqueness then follows from that of the basis-$D$ representation of integers.
\end{proof}

%

\section{Equivalent correlators in non-melonic models}\label{sec:correlators_melon_nonmelonic}

To prove the main lemma of this section we define the swap.
Given two Feynman graphs  $G$ and $H$ with the same number of colours,
their  \textit{swap}  $G \swap{v}{w} H$ at
at $v\in V_0(G)$  and $w\in V_1(H)$ is the graph defined in Figure \ref{fig:swap} (for details, see the original definition of \cite{fullward}, where a different notation is used). The swap
has useful properties, like behaving as the graph theoretical counterpart\footnote{In particular, it is additive in the Gur\u au-degree,
as proven by \cite{fullward} and later also in \cite{Cristofori_Dartois_etal}.} of a
connected sum of topological spaces, but we shall extensively exploit only that
\[
f(  G \swap{v}{w} H  ) = f (G) + f (H) - D.
\label{swap_faces_deficit}
\]

\begin{figure}\centering
\includegraphics[width=.55\textwidth]{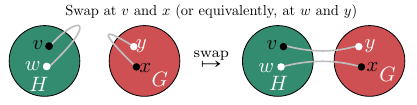}
\caption{The swap for $H$ and $G$ $(D+1)$-coloured graphs (in gray, we represent propagators or $0$ colour).\label{fig:swap}}
\end{figure}
This formula can be proven by noticing that the $D$ faces that contained the
propagator at $v$ and those $D$ faces
that contained the
propagator $w$ merge after applying the swap.\\[-1ex]

Consider now a tensor model $S\inter(T) = \summ i g_i N^{s_i} B_i(T,\bT)$,
which need not be melonic. We prove that for any
(connected) melonic invariant $C$, the large-$N$ moment $m(C)$ depends only on $p(C)$
and not on the graph $C$ itself. That is,
we can drop the whole heavy combinatorics a
graph entails once we had verified that $C$ is melonic, and
in that case $m(C)$ sees only the number of vertices, \textit{even if $S\inter$ is not melonic}.\\[-1ex]


\begin{lemma}\label{thm:bubbles_inequalities_Lemm}
Let $C$ be a connected  melonic graph and $B_{i_1},\ldots,B_{i_n}$ not necessarily melonic, but connected,
and $i_1,\ldots,i_n \in \{1,2,\ldots,m\}$ not necessarily different.
Let us denote the
\dip s of $C$ by close-lying, encircled black-white couples; those of the
red (hatched) graph are arbitrary black-white pairs. With gray solid lines one
denotes only the Wick contractions of $C \dot\cup B_{i_1}\dot\cup B_{i_2}\dot\cup  \cdots \dot\cup  B_{i_n}$ that differ on the two sides of the following inequalities: \\[-4ex]
\begin{subequations}\label{bubbles_inequalities}
\[\label{bubbles_inequalitiesA}
\raisebox{-.50\height}{\includegraphics[width=.417\textwidth]{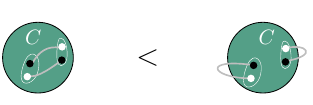}} \phantom{ABc}\\[2ex]\label{bubbles_inequalitiesB}
\\[-7ex]\label{bubbles_inequalitiesC}
\raisebox{-.50\height}{\includegraphics[width=.57\textwidth]{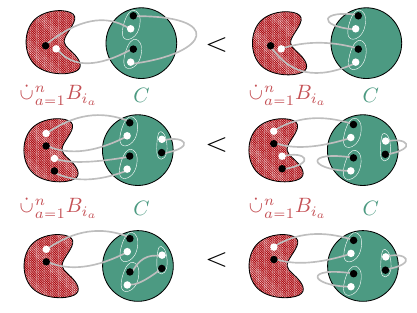}} \\[-8ex]\label{bubbles_inequalitiesD}
\\\notag
\]
which hold after applying $f (\balita)$, that is `after taking the number of faces',
on both sides. All blobs are connected (after Wick contracting; between the enumerated
lines, we display the graphs that are being Wick-contracted).

\end{subequations}
\end{lemma}
Before writing the proof, let us explain, less pictorially, the meaning of the previous equations, respectively:
\begin{enumerate}[(a)]
 \item Let us call the Wick pair of a melon
 that consists
 exclusively of \dip's the \textit{canonical Wick pair}.
 Then any other Wick pairing of a melon
 brings in less faces than the canonical Wick pair.
\item Call $\mathcal{B}=  B_{i_1}\dot\cup B_{i_2}\dot\cup  \cdots \dot\cup  B_{i_n}$. The red (hatched) blob means
a Wick contraction that restricted to $\mathcal{B}$
yields a connected graph. Now suppose that exactly two
vertices of $\mathcal{B}$ are Wick-contracted with two of $C$.
Then Eq. \eqref{bubbles_inequalitiesB}
says that this contraction yields more faces when
the two vertices of $C$ are \dip's than when they are not (all other undepicted
contractions being equal on both sides).

\item Contracting a melon $C$ with the same connected component of $\mathcal{B}$ with more than two propagators
reduces the number of faces.

\item Eq. \eqref{bubbles_inequalitiesD} emphasizes that
Eq. \eqref{bubbles_inequalitiesA} holds also when the melon $C$
is contracted with other graph.

\end{enumerate}

\begin{proof} We prove each inequality independently. In the proof, we refer to the  Wick
contracted $B_{i_1}\dot\cup B_{i_2}\dot\cup  \cdots \dot\cup  B_{i_n}$  as \textit{bulk}.  (The bulk consists of the kidney-like blobs drawn in this Lemma.)

\begin{enumerate}[(a)]
 \item This one has been proven in \cite{TwofoldUniversality}, but we sketch the proof. First, one shows that,
 independent of the anatomy of the graphs $C$, the
 difference $\Delta$ of faces from the leftmost contraction (in gray, at an $i$-coloured dipole)
 minus the number of faces form the rightmost
 contraction (not matching the two dipole vertices) below, is at least $D-2$:
\[
\runter{\includegraphics[width=.5\textwidth]{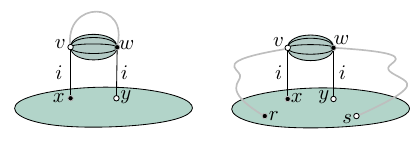}}
\]
Indeed, the differences $\Delta$ of faces (L minus R) is $\Delta=D-2$ if there is a path of alternating
colours $(0,i,0,i,\ldots, 0,i)$ from $x$ to $r$, and is $\Delta=D$ if not. The fact that always $\Delta>0$  forces any faces-maximizing Wick contraction
 to sit at any dipoles of $C$. It is not complicated to prove that this
 maximality is preserved by removing dipoles (and the respective Wick pair that joins their vertices)
and adding dipoles, as far as the new vertices are also Wick contracted (details are in \cite[Lem. 3.4, Rem. 3.5]{TwofoldUniversality}). But this is precisely the property describing \dip s, which finishes the proof

\item
Let $G_L$ and $G_R$ the graphs on the left and the right of Ineq. \eqref{swap_faces_deficit}.
Observe that both are the swap of the hatched blob $\Gamma$
(the same Wick  contraction of $B_{i_1}\dot\cup B_{i_2}\dot\cup  \cdots \dot\cup  B_{i_n}$)
with two different Wick contractions $\pi_L$ and $\pi_R $, respectively, so $\gamma_L:=\pi_L(C)$ and $\gamma_R:=\pi_R(C)$ of $C$.
By Eq. \eqref{swap_faces_deficit},
\[
f(G_\bullet) = f(\Gamma) +  f(\gamma_\bullet) -D  \qquad \bullet=L,R.
\]
Subtraction yields
\[
f(G_R)-f(G_L) =   f(\gamma_R) -f(\gamma_L)
\]
Notice that $\gamma_R$ and $\gamma_L$ are a pair of graphs that appears in
 Ineq. \eqref{bubbles_inequalitiesA}, therefore $f(G_R)-f(G_L) >0$.

\item In this third graph inequality we depicted the uppermost pair of propagators
to stress that we remain in the set of connected Feynman graphs. So we can focus only on the changes
that the swap of the lower-most four vertices in Ineq. \eqref{bubbles_inequalitiesC}.
Observe that the melon $C$ is attached to the bulk in the LHS of the equality
via four propagators. Two cases emerge: either the four vertices are connected
by a line of colour $i$ (for some $i=1,\ldots,D$) and then so happens in the RHS of the same inequality, or not. The affirmative case
is depicted here:
\[ \label{case_mit_i}
\runter{\includegraphics[width=.57\textwidth]{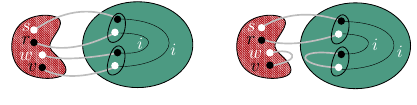}}
\]
Observe that the face of colour $i$ can be shared by the propagators at the
two sets of \dip s on the melon. On the bulk side,
let $k$  ($k=0,\ldots,D$) denote the number of common faces
(with colours $j_1,\ldots,j_k$) that are shared by all four depicted vertices $r,s,v,w$ in bulk.
We depict with a zigzag line the collective of the edges with colours $j_1,\ldots,j_k$
that share (by  assumption) endpoints. The propagator is denoted, as always here, with gray: \vspace{-2ex}
\[
\runter{\includegraphics[width=.135\textwidth]{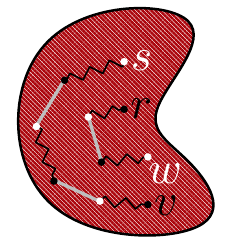}}
\]\vspace{-2ex}

There are still two cases to consider:   \begin{itemize}
                                          \item
If $i\notin \{j_1,\ldots,j_k\}$,
it follows from edge-regularity of the graphs, that the number of
faces implied in the four propagators is $3D-k$ for the rightmost graph in
\eqref{case_mit_i}, while $2D-k-1$ in the graph on the left there.
Their difference being  positive shows  Ineq. \eqref{bubbles_inequalitiesC}.

     \item If $i\in \{j_1,\ldots,j_k\}$,
     the four propagators bring in $3D-k-1$ faces on the right, and
     $2D-k+1$ on the left. Their difference $D-2>0$ shows   Ineq. \eqref{bubbles_inequalitiesC}.
     \end{itemize}

If the anatomy of the melon is \textit{not} as in \eqref{case_mit_i} above,
but rather
\[
\runter{\includegraphics[width=.57\textwidth]{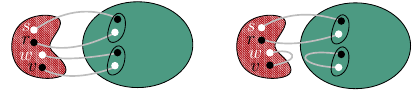}}
\]
then the difference of the number of faces RHS $-$ LHS
is always $D$, independently of the anatomy of the bulk (the analysis
has to be repeated as above, but it is not illuminating),
and Ineq. \eqref{bubbles_inequalitiesC} follows.
\item Finally, Ineq. \eqref{bubbles_inequalitiesD}
holds if and only if it holds after we swap the propagators
in such a way that the graph is disconnected.
Then we just apply Ineq. \eqref{bubbles_inequalitiesA} and the result follows.\qedhere
\end{enumerate}

\end{proof}

Some readers might have skipped the proof and yet wish to know the essence
of melonicity in Lemma \ref{thm:bubbles_inequalities_Lemm}, so we comment on this.
Were $C$ inside
\eqref{bubbles_inequalities} not melonic, then
there are no \dip s and the following double connections could happen:
\[
\runter{\includegraphics[width=.57\textwidth]{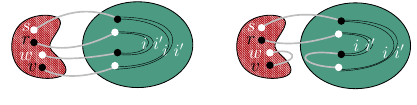}}
\]
In this non-melonic situation there exist cases for which the rightmost graph has less faces than the one on the left. For instance, suppose that there are $k$ faces of colours
$j_1,\ldots,j_k$ joining $s$ with $v$ and $r$ with $w$ above.
Then if $i,i' \in \{j_1,\ldots,j_k\}$, the difference faces${}_R$ $-$ faces${}_L$ is $D-4$, which could be non-positive. So \textit{one} role of melonicity
is to allow at most one colour to connect the two \dip s inside the melon $C$, then
faces${}_R$ $-$ faces${}_L = D-2>0$, as seen above.
\\

We let $\mathscr M_n(B_{i_1},B_{i_2},\ldots, B_{i_n})$ be the set of face-maximizing
Wick pairings that yield a connected Feynman graph out of the invariants $B_{i_1},B_{i_2},\ldots, B_{i_n}$
(cf. \cite{TwofoldUniversality}). Armed with Definition \ref{def:order_melonic}
and Lemma \ref{thm:bubbles_inequalities_Lemm},  we can prove the main claim of this section.

\begin{proposition}\label{thm:ExpMelonEqualPropo}
For melonic $C$ and $\tilde C$, there  is a bijection $\phi$ between the maximal Wick pairings
of a fixed set of interactions with either $C$ or $\tilde C$. That is,
\[ \phi : \mathscr M_{n+1} (C,B_{i_1},B_{i_2},\ldots, B_{i_n}) \to  \mathscr M_{n+1} (\tilde C,B_{i_1},B_{i_2},\ldots, B_{i_n})\]
for any $i_1,\ldots,i_n \in \{1,\ldots,m\}$. In particular, $\phi$ preserves the number of faces.
\end{proposition}

\begin{proof}
For sake of  brevity, in this proof we let \[\mathscr M= \mathscr M_{n+1} (C,B_{i_1},B_{i_2},\ldots, B_{i_n})  \,\, \und \,\,\tilde {\mathscr M}=\mathscr M_{n+1} (\tilde C,B_{i_1},B_{i_2},\ldots, B_{i_n})
\]

We define $\phi$ on a Wick partition in  $\pi \in \Msc$. Let $G=G_\pi$ abbreviate the Feynman graph $\pi$ gives rise to,  $G_\pi=
\pi\big( C\dot\cup [\dcupa ^n B_{i_a}] \big)$.
If we remove $C$ from $G_\pi$  we obtain in general a
disconnected graph $G_\pi\setminus C$,
whose number of connected components will be denoted by
$d_C$ ($1\leq d_C \leq p(C)$). To keep the proof concise,
let us call \textit{bulk} the graph $G_\pi\setminus C$, after removing also the propagators that
were attached to $C$.\\

Thanks to Lemma \ref{thm:bubbles_inequalities_Lemm} we know how the
several connected components of the bulk are attached to $C$:  there are
exactly $2d_C $ external legs, with two per connected component of course.
Indeed, first
it obviously
cannot have less by vertex-bipartiteness,
so it remains to see that it cannot have more than $2d_C$.
Suppose the  contrary, namely that the bulk has a number
of external legs larger than $ 2d_C$. Then, there exist at least one connected
component of the bulk that is attached to $C$ via at least four propagators. But then from
Ineq. \ref{thm:bubbles_inequalities_Lemm} we obtain a contradiction with the
maximality of $\pi$.\\

In $C$, the vertices might have another labels
that do not match the canonical one in Definition \ref{def:order_melonic},
so let us relabel to correct this. Let $\tau$ be the unique bijective map $\tau: \{0,1,\ldots, 2p-1\} \to  \{0,1,\ldots, 2p-1\}$ that relabels the original vertices $V(C)$
as $\lambda^{\can}(C)$.
In particular, $\tau=(\tau_0,\tau_1)$
splits into two permutations $\tau_0$ of the
black and $\tau_1$ of the white vertices
(or $\tau_0 \in \Sym \{0,2,\ldots,2p-2\},\tau_1 \in \Sym \{1,3,\ldots, 2p-1\}$).
Similarly for $\tilde C$, we correct its labels with a bijection
$\tilde \tau: \{0,1,\ldots, 2p-1\} \to  \{0,1,\ldots, 2p-1\}$ (observe that $p(C)=p(\tilde C)$ by assumption), so $\tilde \tau=(\tilde\tau_0,\tilde\tau_1)$ maps
the original labels of $\tilde C$ to
$\lambda^\can(\tilde C)$.\\

\begin{figure}[htb!]\centering
\includegraphics[width=.79\textwidth]{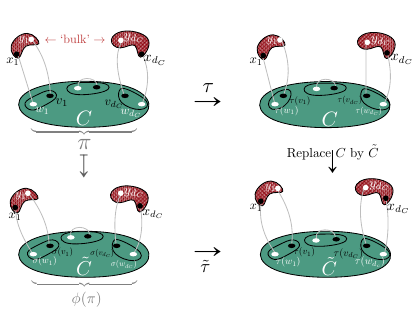}
\caption{In gray, the map $\phi(\pi)$ that
yields a new Wick contraction defined by permuting $V(\tilde C)$
by $\sigma$ and pairing in such a way, that the diagram commutes (i.e. $\tilde \tau \circ \sigma =  \tau$).
The map that consists in the replacement of $C$ by $\tilde C$ is possible thanks to Definition \ref{def:order_melonic} and Claim \ref{thm:unicity_max_claim}. \label{fig:proof}}
\end{figure}

Let\footnote{We clarify
that the notation here is functional, so $\tilde\tau\circ \tau$ means
`first $\tau$ and then $\tilde \tau$' which differs from
the usual multiplicative notation for permutations.} $\sigma= (\tilde\tau )\inv\circ \tau$, and
$\phi(\pi) = \sigma^*(\pi)$ denote precomposition by $\sigma$. Explicitly,
$\phi(\pi)$ consists of the following propagators
\[ \label{phipi_explicit}
\phi(\pi) =
            \Big\{  (\sigma(x_a), w_a) ,\,
              (v_a , \sigma(y_a) ),\,
               (\sigma(x_a), \sigma(y_a ) )  \Big\}_{ a=1,\ldots,d_C}  \cup
           \pi|_{\textrm{bulk}}.
\]
where  the very last set means
\[\pi|_{\textrm{bulk}} =
\big\{(x,y) \in \pi \,:\, x,y \in V( \dcupa^n B_{i_a} ) \setminus V(C)\big\}. \]
We now verify that $\phi$ satisfies what we claimed:
\begin{itemize}
 \itemb \textit{Well-definedness}.  Since
 $\sigma=(\sigma_0,\sigma_1)$ splits as a couple of permutations
  \[ \sigma_0  & = (\tilde \tau_0)\inv \circ  \tau_0
  \in \Sym \{0,2,\ldots,2p-2\}  \\ \sigma_1 &=  (\tilde \tau_1)\inv
  \circ  \tau_1 \in \Sym \{1,3,\ldots, 2p-1\} \]
   $V(C)$ is bijectively replaced via $\sigma$ by
 $V(\tilde C)$, keeping bipartiteness, so
 $\phi(\pi)$ is a Wick contraction of $\dcupa^n B_{i_a} \dot\cup  \tilde C$,
 since $\pi$ was a Wick contraction of
 $\dcupa^n B_{i_a} \dot\cup   C$.

 \itemb \textit{Maximality is preserved by $\phi$ }. Let
 $\mathcal B=(\dcupa^n B_{i_a} )\dot\cup   C$
 and
 $\tilde {\mathcal B}=(\dcupa^n B_{i_a} )\dot\cup \tilde  C$.
 Define $\tau^*\pi$ as in Eq. \eqref{phipi_explicit},  replacing $\sigma$ by $\tau$ there.
 Notice that this just relabels vertices, but
 the number of faces of $\pi(\mathcal B)$
 is the same. The main point now is that a connected
 melon has always the following maximal number of faces
 \[ \label{numfacesmelon1} \qquad
 \max \{  f (G_\Pi ) :  G_\Pi=\Pi (C),\Pi \text{ Wick contraction}  \} = p(C) \times (D-1) + 1.
\]
and maximality is achieved only when one Wick contracts the \dip s \cite{TwofoldUniversality}.
More important than the exact number is the fact that
\[\label{numfacesmelon2}
\max_{\Pi \, \Wick } f \big( \Pi(C)\big ) = \max_{\Pi \, \Wick } f\big( \Pi(\tilde C)\big ) ,
\]
since $p(C)=p(\tilde C)$ by assumption. Now
the next important fact is that $\pi(\mathcal B)$ has been obtained as
$d_C$  consecutive swaps of the $d_C$ connected components $\mathfrak b_1,\ldots, \mathfrak b_{d_C}$
of the bulk
(red or hatched blobs in Fig. \ref{fig:proof}) and the melon $C$.
Thus, we can compute faces everywhere using Eq.
\eqref{swap_faces_deficit}, which holds also for non-melonic graphs. Indeed,
observe that if we abbreviate  $\tilde G_{\phi(\pi)}= \phi(\pi) ({\mathcal B})$,
and if $\pi_{\max}(C)$ denotes the Wick contraction
defined by pairing all the \dip's, then
\[\notag
f(\tilde G_{\phi(\pi)} ) &= f(\pi_{\max} C) +  \bigg(
\sum_{a=1}^{d_C}\,
f(\mathfrak b_a)  \bigg) - d_C \times D
\\ \label{equal_faces_phipi}
&= \big [p(C) \times (D-1) + 1 \big] +  \bigg(
\sum_{a=1}^{d_C}\,
f(\mathfrak b_a)  \bigg) - d_C \times D  =  f(G_{\pi}),
\]
by Eqs. \eqref{numfacesmelon1} and
\eqref{numfacesmelon2}. Hence not only does
$\phi$ map maximal to maximal Wick contraction,
but it also preserves also the number of faces.

\itemb \textit{Bijectivity of $\phi$ }. From $ (h\circ g)^*=g^*\circ h^*$ satisfied by precomposition,
one can easily prove that the
inverse of $\phi$ is
$\psi$, defined on an element $\tilde \pi\in\tilde{ \mathscr M}$ by
$\psi(\tilde \pi)=(\tau\inv \circ \tilde\tau)^*(\tilde\pi).  $ \qedhere
\end{itemize}

\end{proof}
In summary, in this proof,
Lemma \ref{thm:bubbles_inequalities_Lemm} determines
the form of contracting the melon with (possibly)
non-melonic interactions. It is in fact not
too complicated to construct a map
\[ \mathscr M_{n+1} (C,B_{i_1},B_{i_2},\ldots, B_{i_n}) \to  \mathscr M_{n+1} (\tilde C,B_{i_1},B_{i_2},\ldots, B_{i_n})\nonumber \]
but the fact that the map $\phi$ constructed above is a bijection is merit of
Definition \ref{def:order_melonic}.

\begin{corollary}\label{thm:EC_isEtildeC_coro}
For $C$ and $\tilde C$ melonic, the leading order of
the moments $\E  _{g} [C]$
 and $\E _{g}  [{\tilde C}] $
coincides, when they converge.
\end{corollary}

\begin{proof}
That bijection in Proposition
\ref{thm:ExpMelonEqualPropo} preserves maximality as well as the precise
number of faces means that connected Feynman graphs
with vertex sets $\{C,B_{i_1},B_{i_2},$ $\ldots,B_{i_n} \}$
and  $\{\tilde C,B_{i_1},B_{i_2},$ $\ldots,B_{i_n} \}$ for any
 $i_1,\ldots,i_n\in \{1,\ldots,m\}$
 are indistinguishable at leading order. Since any disconnected
 Feynman graph with any of the two previous sets of interactions
 has the melon in exactly one connected component (while the rest
 consists of Feynman graphs in the remaining $B_i$'s) then
 $\E  _{g} [\tilde C] - \E  _{g} [  C]= \mathcal Z\inv
 \int [\tilde C(T) -C(T) ] \ee^{\summ{i} g_iB_i  } \dif\mu_0(T)$ is indeed of a sub-leading order.
\end{proof}

\section{Reduced single-trace Schwinger-Dyson}\label{sec:SDE_redux}

The tensor Schwinger-Dyson or loop equations require the concept of
the union  \[ B \underset{x,w}{\cup} B' \] of two graphs $B,B'$
at two given vertices $x\in V(B) $ and $w \in V(B')$ of opposite parity,
e.g. $x$ is black (even) and $w$ white (odd).
We feel obliged to provide the full definition
but also a shortcut (see Fig.
\ref{fig:ex_union_Graphs}). The vertex and edge sets of the union of graphs are given by
\[ V(B \underset{x,w}{\cup} B') & = [V(B) \cup V(B')] \setminus \{ x,w\} \\
 E (B \underset{x,w}{\cup} B') & = E(B) \setminus \{  ( x,t_c ) : c=1,\ldots,D  \}
 \\
&\cup     E(B) \setminus \{  ( s_c,w  ) : c=1,\ldots,D \} 
\notag
 \\[1ex]
 &
\cup \{ (s_c,t_c)  : c=1,\ldots, D\} \notag
\]
where $s_c$ is the black (resp. $t_c$ is the white) vertex connected by the single colour-$c$
edge to $w$ (resp. to $x$), as depicted in Figure \ref{fig:faces_after_union}. In summary, the union of two graphs at two vertices
`has the edge sets of both graphs except those attached to
either of the given points; one removes these, breaking $D$ edges
into half-edges for each of the two vertices,
and weld the $2D$ half-edges colour-wise'.

\begin{figure}
\runter{\includegraphics[width=.58\textwidth]{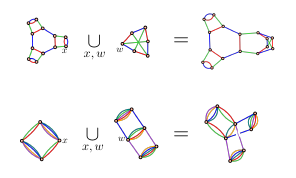}}
 \caption{The union of graphs exemplified for $D=3$ (above)
 and $D=6$ (below). The vertex parity is not depicted, as the bipartiteness of the graphs is obvious.\label{fig:ex_union_Graphs}}
\end{figure}

\begin{lemma}\label{thm:union_Faces_sum_lemma}
For any connected $D$-coloured graphs $B,C$ and $B'$, the following holds:
\begin{enumerate}[(a)]

 \item
 For Wick contractions $\pi$ of $B$
 and $\pi'$ of $B'$, let $v$ and $y$ be the only vertices determined by
 $(x,y)\in\pi$ and $(v,w) \in \pi' $. If
\[\Pi = (\pi \setminus \{x,y\} )\cup (\pi' \setminus \{v,w\}  )\cup \{ (v,y) \}  \label{bigPi}\]
then
\[ \label{faces_deficit_union}
f \Big[ \Pi\big ( B \underset{x,w}{\cup} B'\big)  \Big]  = f  [\pi(B)] + f [\pi'(B')] -D.
 \]

  \item For melonic $C$ and
vertices of opposite parity $x\in V(B)$ and $w\in V(C)$ (say, $x$ black and $w$ white)
there exists a bijection $\gamma_{x,w}$
of maximal Wick contractions:
\[ \gamma_{x,w} :
\mathscr M_1 (B)\times
\mathscr M_1 (C) \to
 \mathscr M_1 \big( B \underset{x,w}{\cup} C \big).
\]

\end{enumerate}

\end{lemma}
\begin{proof} For (a), observe that since $B$ and $B'$
share no vertices nor edges,  $x$ lies in $D$
faces and $w$ lies in a disjoint set of $D$ faces.
When we perform the union at $x$ and $w$, these two
disjoint set of faces merge into exactly $D$.
Indeed, all other edges not containing these two points
were respected, and connect the
edges with colours in $\{1,\ldots,D\}$ and, according to \eqref{bigPi}, $\Pi$ has
a $0$-th colour (propagator) that
closes such $D$ faces, which we depict next only using colour-$c$ faces in Figure \ref{fig:faces_after_union}.
\begin{figure}
\centering
\runter{\includegraphics[width=.41\textwidth]{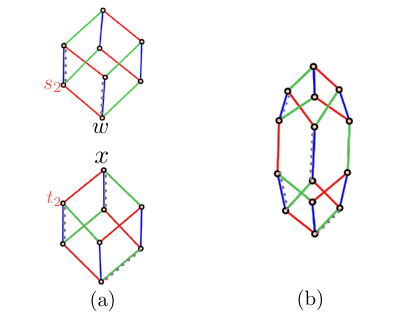}}
\runter{\includegraphics[width=.41\textwidth]{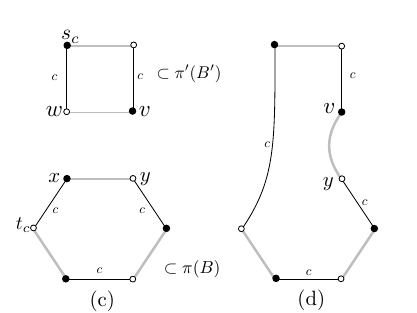}} \hspace{-3ex}$\subset \Pi(B \underset{x,w}{\cup} B')$
\caption{
(a) Two graphs $B$ below and $B'$ above,
and their
vertices $x$ and $w$, (b) shows $ B \cup_{x,w} B'$.
In dashed gray, we depict a selection of edges from a Wick contraction that
plays a role in (c) and (d) (in solid gray). In particular,
(c) depicts the colour-$c$ faces [in (a) and (b) concretely in red, or colour 2] of $B$ and $B'$. In (d)
these faces merge for the Wick contraction defined Eq. \eqref{bigPi}.\label{fig:faces_after_union}}
\end{figure}
Since this situation is independent for each colour,
after the union of graphs we get a deficit of $D$ faces. \\

Concerning (b), let $\gamma_{x,w}(\pi,\pi')= (\pi \setminus \{x,y\} )\cup (\pi' \setminus \{v,w\}  )\cup \{ (v,y) \}  $ as above.
Then this is a Wick contraction indeed, and
it is maximal since
Eq. \eqref{faces_deficit_union} implies that  $\gamma_{x,w}(\pi,\pi')$
is maximal if and only if $\pi$ and $\pi'$, as we assumed, are.
The inverse map is as follows. Given $\Pi \in \mathscr M_1( B{\cup}_{x,w} C ) $,
 consider the sets
 $V(\check B)=V(B)\setminus  \{   x \} 
$ and $ V(\check C)=V(C)\setminus   \{   w \}$.
Since $C$ is melonic, there are still
$2p-2$ \dip s $\subset V(\check C)$, so these must be paired by $\Pi$
and we can be sure that $\Pi$ does not pair vertices of $C$ with $B$.
Then the inverse of $\gamma_{x,w}$ is $\Pi \mapsto (\pi,\pi')$, whose components are given by
\[ \notag
 \pi=\Pi|_{V(\check B) } \cup \{(x,y)\} \quad\und\quad \pi'=\Pi|_{V(\check C)} \cup   \{(v,w)\},
\]
and maximality follows again from
Eq. \eqref{faces_deficit_union}.\qedhere
%
\end{proof}
\begin{minipage}{.7\textwidth}
\begin{example}
It is interesting to see that if $B'$ above is not melonic,
then  $\gamma_{x,w} :
\mathscr M_1 (B)\times
\mathscr M_1 (B') \to
 \mathscr M_1 \big( B{\cup}_{{x,w}} B' \big)$
is only an injection. If one tries to invert this, $[\#  \mathscr M_1(K_{3,3})] ^2=9$,
while $  \#\mathscr M_1 \big( K_{3,3} {\cup}_{{x,w}}  K_{3,3} \big) =10$.
The culprit is drawn on the right.
\end{example}
\end{minipage}~
\begin{minipage}{.26\textwidth}
\[\notag
\includegraphics[width=11ex]{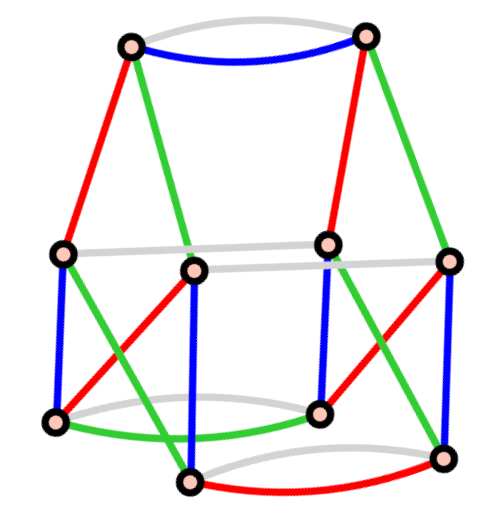}
\]
\end{minipage}
\vspace{-2ex}
\begin{proposition}\label{thm:MC_MCtilde_mitBs_Propo}
For connected, melonic $D$-coloured $C$ and $\tilde C$ with $p(C)=p(\tilde C)$ and $x\in V(C)$ and $y \in V(\tilde C)$, $z\in V(B_j)$,
with $x,y$ having the opposite parity to $z$, then there is a bijection $\alpha$ (that depends on the data $C,\tilde C, B_j, y,x$)
\[
\alpha : \mathscr M_{1+n} \big(C \underset{x,z}{\cup} B_j,  B_{i_1},  \ldots, B_{i_n} \big)
\to
\mathscr M_{1+n} \big(\tilde C \underset{y,z}{\cup} B_j,  B_{i_1},  \ldots, B_{i_n}\big ).
\]
\end{proposition}

\begin{proof}
Unlike the case in which $B_{i_1},  \ldots, B_{i_n} $ are all melonic, we do not have now a tree-like structure.
What we do still have is the fact that $\Pi \in
\mathscr M_{1+n} \big(C {\cup}_{{x,z}} B_j,  B_{i_1},  \ldots, B_{i_n} \big) $
is obtained by applying a finite sequence  of (say, $d$)
swaps $\mathsf s_1,\ldots, \mathsf s_d$ ---which we shall remember---from elements in
\[
(\pi_0,\ldots,\pi_n) \in
\mathscr M_1 (C \underset{{x,z}}{\cup} B_j)  \times
\mathscr M_1 (  B_{i_1}) \times
\cdots
\mathscr M_1 (  B_{i_n}). \label{irdwwas}
\]
But then there is a bijection
\[
\beta: \mathscr M_1 (C \underset{x,z}{\cup} B_j) \stackrel{ \text{\tiny Lem. \ref{thm:union_Faces_sum_lemma} }}{\to }
\mathscr M_1 (C) \times \mathscr M_1 ( B_j)    \stackrel{ \text{\tiny Prop. \ref{thm:ExpMelonEqualPropo} }}{\to }
\mathscr M_1 (\tilde C) \times \mathscr M_1 ( B_j)  \stackrel{ \text{\tiny Lem \ref{thm:union_Faces_sum_lemma} }}{\to}
\mathscr M_1 (\tilde C \underset{x,z}{\cup}  B_j), \notag
\]
where the middle bijection  $\phi \times \mathrm{id}$ (with $\phi$ being the case $n=0$ in Prop. \ref{thm:ExpMelonEqualPropo})
exists because $C$ and $\tilde C$ are melonic, an due to $p(C)=p(\tilde C)$,
by assumption.
Then we are able to replace $C\to \tilde C$, $x\to y$
and $\pi_0$ by $\beta( \pi_0)$ in
\eqref{irdwwas} and swap via $\mathsf s_1,\ldots, \mathsf s_d$
the $n+1$-tuple in ($\beta(\pi_0), \pi_1,\ldots,\pi_n$)
in the same way that took us from
 $\pi_0,\ldots,\pi_n$ to $\Pi$, to obtain a new
 Wick contraction $\Pi'$ of
 $
 C {\cup}_{{x,z}} B_j \dot\cup B_{i_1} \dot\cup  \ldots  \dot\cup B_{i_n} $.
 This $\Pi'$ is maximal, since
 ($\beta(\pi_0), \pi_1,\ldots,\pi_n$) $\to \Pi'$ is obtained via
  $\mathsf s_1,\ldots, \mathsf s_d$
  and  leads then to exactly the same change in the number of faces
as $(\pi_0,\ldots,\pi_n) \to \Pi$.
Then we can set $\alpha(\Pi)=\Pi'$. In particular
$\alpha$ is bijective, since $\beta$ is invertible.
\end{proof}

We now arrive to our Schwinger-Dyson equations for tensor models.
These appeared first in \cite{GurauVirasoro}
as a set of operators that annihilates the partition function
and obeys a generalization of the (half-)Witt Algebra. Since this
has been done elsewhere and the procedure is well-known\footnote{An example of
tensor model with Schwinger-Dyson equations that
became more intricate is presented in \cite{SDE} for a
quartic tensor model inspired by the Kontsevich and Grosse-Wulkenhaar models \cite{gw12}.}, a sketchy path
to derive these equations suffices here. For a tensor model with interactions $ S\inter(T,\bT) $
is to start with the identity
$\int_{[\C^{N}]^{\otimes D}} \dif \big( f(T)  \ee^{-N^{D-1} T\cdot \bar T
- S\inter(T,\bT) }  \big) \dif T =0 $ for $f(T)$, which need not be an invariant.
In fact a useful type of function $f(T)$ is rather like
\begin{subequations}\label{bubble_removal}%
\[
\sum_{ \substack{ \gris{i_1,}i_{2},i_{3}, j_1, \gris{j_{2},} \\ j_{3},k_1,k_{2}\gris{,k_{3}}=1}}^N%
\quad&\gris{\bar{T}_{i_{1},k_{2},j_{3}}}
T_{j_1,j_{2},j_{3}}
\bar{T}_{j_{1},i_{2},k_{3}}& \leftrightarrow &&
\runter{\includegraphics[width=.126\textwidth]{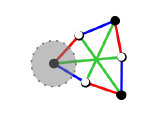}} \\[-2ex]
\cdot & {T_{i_1,i_{2},i_{3}} } T_{k_1,k_{2},k_{3}} \bar{T}_{k_{1},j_{2},i_{3}}  &&& \notag
\\
\sum_{ \substack{  \gris{i_1,}i_{2},i_{3},i_{4},j_1,\gris{j_{2},}j_{3},j_{4}, \\k_1,k_{2},\gris{k_{3},k_{4}}=1}}^N%
\quad &T_{j_1,j_{2},j_{3},j_{4}}\gris{\bar{T}_{i_{1},j_{2},k_{3},k_{4}} }  T_{i_1,i_{2},i_{3},i_{4}} %
&\leftrightarrow&&\runter{
\includegraphics[width=.16\textwidth]{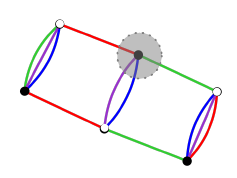}} \\[-2ex] \notag
\cdot & \bar{T}_{j_{1},k_{2},j_{3},j_{4}} T_{k_1,k_{2},k_{3},k_{4}} \bar{T}_{k_{1},i_{2},i_{3},i_{4}} &&&
\\
\sum_{\substack{ i_1,i_{2},i_{3},i_{4} ,\gris{j_1,{j_{2},}j_{3},j_{4},}\\k_1,k_{2}, k_{3},k_{4}=1}}^N%
&\gris{T_{j_1,j_{2},j_{3},j_{4}}} {\bar{T}_{i_{1},j_{2},k_{3},k_{4}} } T_{i_1,i_{2},i_{3},i_{4}} %
&\leftrightarrow&&\runter{
\includegraphics[width=.16\textwidth]{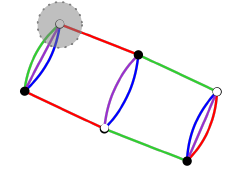}} \\[-2ex] \notag
\cdot &   \bar{T}_{j_{1},k_{2},j_{3},j_{4}} T_{k_1,k_{2},k_{3},k_{4}} \bar{T}_{k_{1},i_{2},i_{3},i_{4}} &&&
\]\end{subequations}

where the gray terms are meant to be absent, and
the vertices encircled on the right excised.
Due to the symmetry in the indices (and the fact that these
are real-valued invariants), the upper most invariant/graph does
not depend on the dropped factor/vertex, but the lack of symmetry
in the last two graphs shows that it in general does.
This means that, in general, the Schwinger-Dyson equations
depend on a graph $C$ and a vertex $x\in V(C)$.  After making sense of all the differential operators (e.g. following \cite{TensorBootstrapProgram} or \cite{openbubble})
one obtains for fixed $x\in V_0(C)$,
\[
 N^{D-1}\E [C] + \summ j \sum_{\substack{w\in V_1(B_j)} } g_j N^{s_j}\E [{ C
\underset{x,w}{\cup}
B_j   }]
= \sum_{ y \in V_1(C) }  N^{\#E(x,y) }
\E [ C\setminus \{x,y\}|_{\text{weld}}  ] \tag{SDE${}_{C,x}$} \label{SDE}
\]
where $E(x,y)$ denotes the set of edges shared by $x$ and $y$.
If $x, y \in V(C)$ have, as above, opposite parity,
then $C\setminus\{x,y\}|_{\text{weld}}$
is the graph obtained by (a) removing from $C$ the vertices
$x$ and $y$, (b) breaking the $D$ edges attached to each into a half
that remains fixed to the vertices that were not $x$ nor $y$,
(c) welding the halves colour-wise into a $c$-coloured edge.\\
\begin{table}[H]
\begin{tabular}{c|c|c|c|c} $D$ &
$C$   & $(x,y)$ &  $C\setminus\{x,y\}|_{\text{weld}}$  &
$N^{\#E(x,y)} $\\[.51ex] \hline
3 &
\runter{
\includegraphics[width=.13\textwidth]{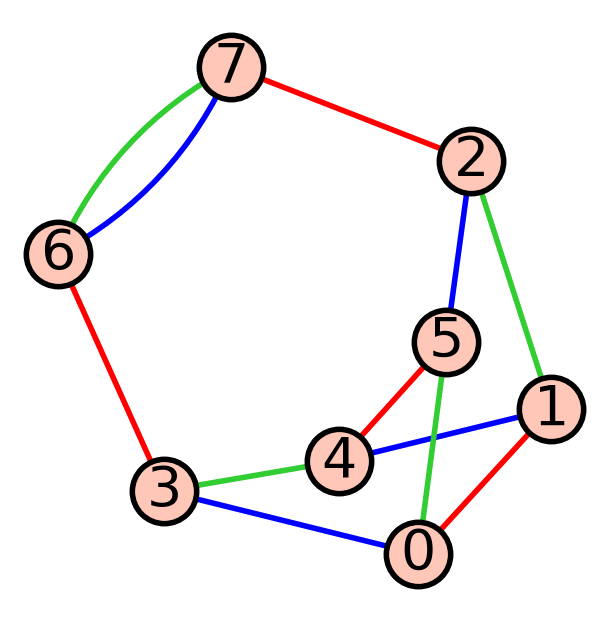}}& (6,7) &
\runter{
\includegraphics[width=.081\textwidth]{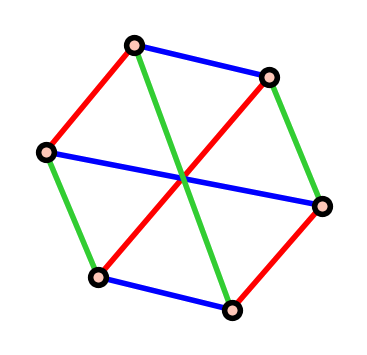}}
& $N^2$ \\ 6 &
\runter{
\includegraphics[width=.136\textwidth]{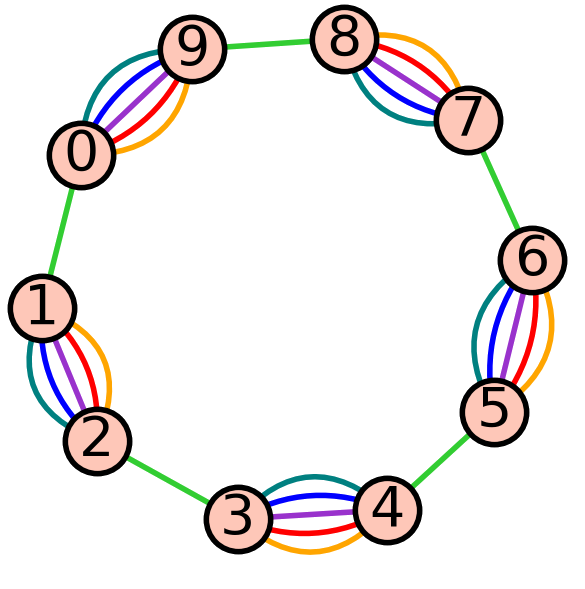}}& (0,5) &
\runter{
\includegraphics[width=.071\textwidth]{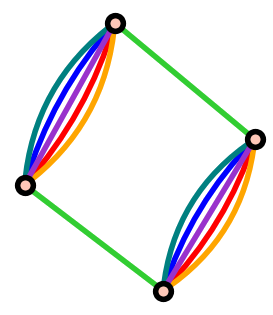}}
\runter{
\includegraphics[width=.071\textwidth]{ex_weld4.png}}
& $N^0$ \\ 6 &
idem & (0,9) &
\runter{
\includegraphics[width=.1\textwidth]{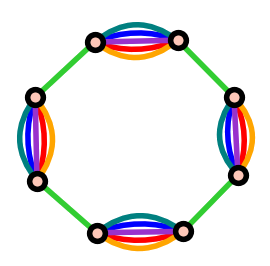}} & $N^{5}$
\end{tabular}
\caption{Examples of the breaking and welding $C\setminus\{x,y\}|_{\text{weld}}$
of graphs $C$.\label{tab:weld}}
\end{table}
Since we are not only interested in the graph, but also
in the set $E(x,y)$ of edges that \textit{completely} disappear while forming
$C\setminus\{x,y\}|_{\text{weld}}$, let us illustrate both concepts simultaneously
in Table \ref{tab:weld}. Observe there, that
break-and-weld can yield disconnected graphs. As correlators of
disconnected graphs scale with higher powers of $N$, in order to obtain
the large-$N$ limit of the Schwinger-Dyson equations, we do not only need
to search for dipoles (yielding $N^{D-1}$ as in the first and last example of Tab. \ref{tab:weld})
to break and weld so that the RHS scales as
the LHS's $N^{D-1}\E[C]$, but also take into account the
disconnected components. Generally
neither the $N^{s_i}$ factor accompanying the operators $B_i$ nor the global factor $N^{q_C}$
 Eq. \eqref{SDE} has to be divided by, so that it is finite at large-$N$,
are algorithmically known.
All $s_i$ and $q_C$ have to be determined. For melonic models $s_i=D-1$
and then $q_C$ is fixed by an algorithm given in \cite{TensorBootstrapProgram},
but as far as the determination of these (usually, integer) parameters is concerned,
this is an interesting problem.
\begin{theorem}\label{thm:SDE_are_equal_thm}
For melonic, connected $D$-coloured graphs $C$ and $\tilde C$ with the same number of vertices, let
$x \in V_0(C)$ and $\tilde x \in V_0(\tilde C)$. Then the LHS of the  Schwinger-Dyson Equation
\ref{SDE} equals
the LHS of SDE${}_{\tilde C,\tilde x}$ (with tilde input data) in the large-$N$ limit.
Concretely, the following leading orders  $(\mathrm{L.O.})$  in $N$ coincide for all $j=1,\ldots,m$ and all $w \in V_1(B_j)$,
\[ \notag
\mathrm{L.O } \, \big\{ \E [C] \big\}=
\mathrm{L.O }\,\big\{
 \E[\tilde C]\big\} \quad\und\quad
\mathrm{L.O }\, \big\{\E [ C
\underset{x,w}{\cup}
B_j   ] \big\} =
\mathrm{L.O }\,\big\{ \E [\tilde  C
\underset{\tilde x,w}{\cup}
B_j   ]\big\} 
\]
whenever any of the two converges,
independently of whether $B_j$ is melonic.
Then their RHSes must agree too, but not trivially, and they provide
relations among the connected and disconnected moments.
\end{theorem}

\begin{proof}
The first equality has been proven by Corollary \ref{thm:EC_isEtildeC_coro},
and for the second equality we use now Proposition \ref{thm:MC_MCtilde_mitBs_Propo}. This states that the maximum of faces
of connected Feynman graphs that are Wick contractions of
$B_{i_i},\ldots,B_{i_n},  C \cup_{ x,w}
B_j$ and $B_{i_i},\ldots,B_{i_n}, \tilde  C \cup_{ \tilde x,w}
B_j $ are the same for $i_a=1,\ldots, m$,  for any $j$ and any $w$. Thus the L.O.
of the cumulants of both $ C \cup_{ x,w}
B_j$  and $  \tilde  C \cup_{ \tilde x,w}
B_j$ agrees at any order in perturbation
theory  (i.e. here, for any $n\in \Z_{\geq 0}$).
In the same way that we passed from cumulants to moments/correlators in Corollary
in \ref{thm:EC_isEtildeC_coro}, we obtain the second equality.
\end{proof}


\section{Simplification of the positive semidefinite matrix in the tensors}\label{sec:BootsMatrices}

\subsection{Positivity bootstrapping matrix ensembles}
This paragraph describes with the minimal non-trivial
amount of information the bootstrap strategy \cite{Lin} for matrix models.
For the ensemble $\cH_N$ of hermitian matrices $N\times N$
with a polynomial $V_g(M)$ in $M$ with coefficients parametrized by  $g\in \R$, let
\[ \notag
\mathcal M(g) = \begin{bmatrix}
                 1 & m_2 & m_4 & m_6 &\cdots\\
                 m_2 & m_4 & m_6 & m_8&\cdots\\
                  m_4 & m_6 & m_8 & m_{10} &\cdots\\
                  m_6 & m_8 & m_{10} & m_{12}&\cdots\\
                  \vdots & \vdots & \vdots & \vdots & \ddots
                \end{bmatrix} \qquad m_{2p}(g) = \lim_{N \to \infty } \frac{1}{Z}\int_{\cH_N} \frac{1}N \Tr M^{2p} \ee^{-N \Tr V_g(M) } \dif M,
\]
being $\dif M$ the (normalized)
Lebesgue measure on $\R^{N^2}\cong\cH_N$ and
$Z=\int_{\cH_N} \ee^{-NV_g(M)}\dif M$.
It is easy to prove the
positivity semi-definiteness of the matrix $\mathcal M(g)$ for all $g$.\\

Indeed, observing that $X=  \sum_i z_i M^{2i}$ satisfies $\Tr (X^*X) \geq 0$
and for any
choice of finitely many non-zero $z_i\in \C$,
so does its expectation
\[  \int_{\cH_N} \Tr (X^*X) \ee^{-N \Tr V_g(M) }  \dif M \geq0\qquad \forall g \,\und \forall z_i.\]
Expanding $X$ it follows that $\mathcal M(g)\succeq 0$.
In the concrete case of  $V_g(M) = \frac{1}{2} M^2+ \frac g4   M^4$,
the Schwinger-Dyson equation for $m_{2p}$
can be solved\footnote{For instance $
m_{4} = (1-m_{2})/g$, $
m_{6}= (2m_{2}-m_{4})/g$,
$m_{8} =  (2m_{4}+m_{2}^2-m_{6})/g$, etc. } in terms of $m_{2p-2}$ and the recursion yields
a parametrization of $\mathcal M$ by $m_2$ and $g$ alone. Then $\mathcal M(g,m_2) \succeq 0$ yields constraints for
the values of the function $m_2=m_2(g)$, which eventually converge
to the solution. Already with two submatrices of sizes $3\times 3$
one finds in $\sim$10 seconds the constraints shown in Figure \ref{fig:BootrMM}.

                \begin{figure}\[
\runter{\includegraphics[width=.3\textwidth]{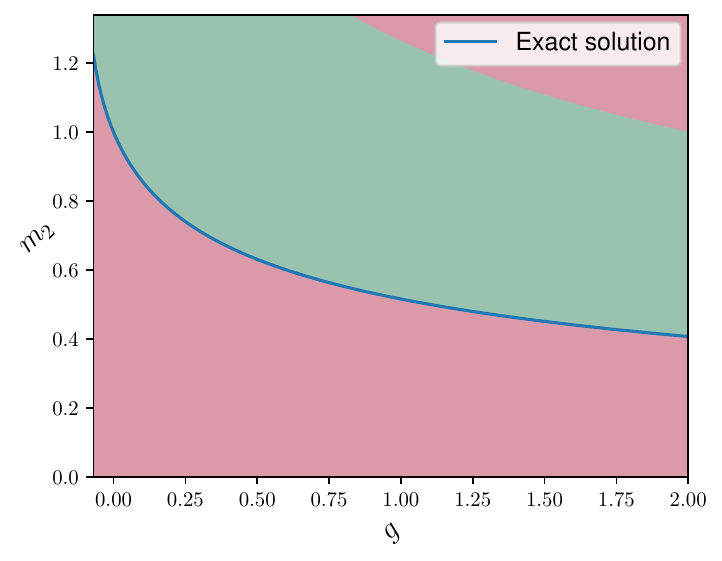}} & \text{\huge$+$}\!\!\!\runter{
\includegraphics[width=.3\textwidth]{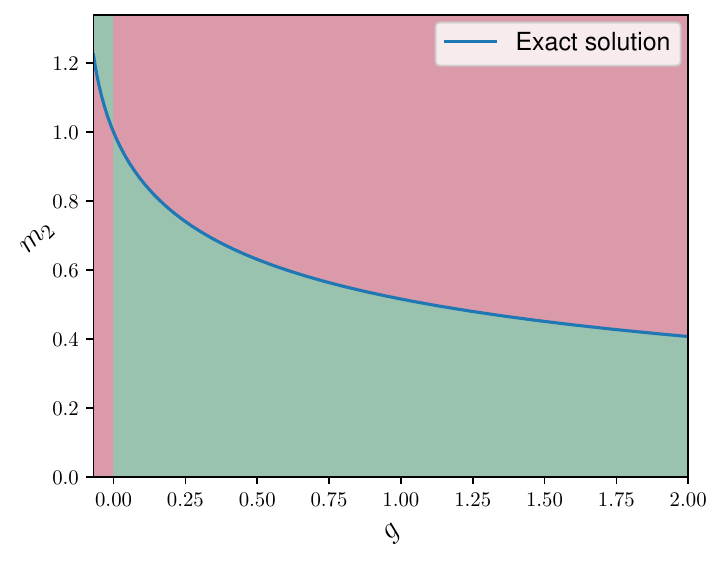}}
\text{\huge$=$}\!\!
\runter{\includegraphics[width=.32\textwidth]{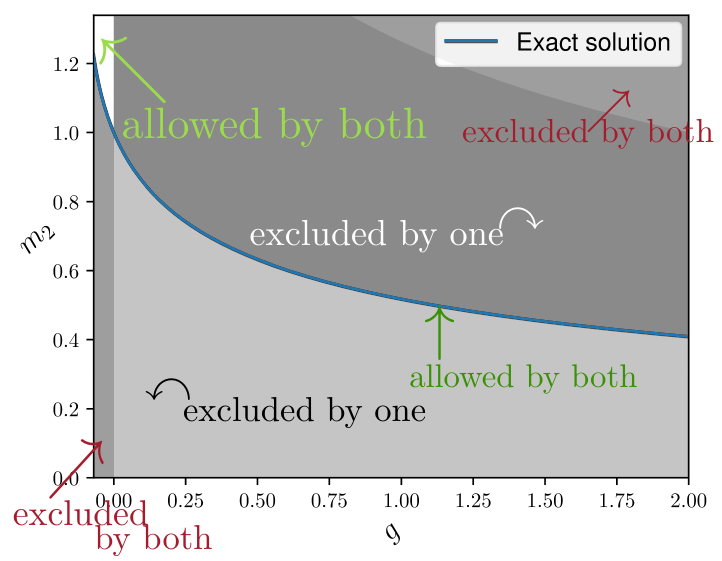}} \notag \\
\det \begin{bmatrix}
                   m_2 & m_4 & m_6  \\
                 m_4 & m_6 & m_8 \\
                   m_6 & m_8 & m_{10}
                \end{bmatrix} \geq 0 &\qquad
  \det
\begin{bmatrix}
   1 & m_2 & m_4  \\
                 m_2 & m_4 & m_6  \\
                  m_4 & m_6 & m_8
                \end{bmatrix}\geq 0 \qquad \qquad\,\,\text{both determinants}\notag
\]\caption{The two first graphs show
the constraints imposed by the determinant's positivity below it with green (red is where the determinant is negative).
The last graph is the superposition of both constraints. The moments $m_{2p}$ of the quartic
ensemble
are expressed as function of $m_2$ and $g$ by means of the Schwinger-Dyson equations.
The rightmost graph allows for $g\geq 0$ only a
thin stripe around the exact solution (blue).
\label{fig:BootrMM}}
                \end{figure}

\subsection{positive semidefinite matrices for large-$N$ tensor models}\label{sec:PosSemiDef_Tensors}

There are several ways to arrange the moments of random tensors in a
positive semidefinite  matrix (\psdm). Let us sketch two useful
proposals, assume that in the potential $S\inter(T,\bT) =
\sum_{j=1}^m g_j N^{s_j} B_j(B,\bT)
$ each $B_j$, is real valued, and abbreviate the list of real
couplings as $g=\{g_j\}_{j=1}^m $.
\begin{itemize}
 \itemb
The matrix that was used in \cite{TensorBootstrapProgram}
can be constructed as follows, starting from
an arbitrary list of connected
$D$-coloured graphs  $\{A_\alpha\}_\alpha$, and assume that $A_\alpha(T,\bT)\in\R$ for each $\alpha$.  One lets\vspace{-1ex}
\[
X_z(T,\bT) & := \sum_\alpha z_\alpha   \partial A_\alpha(T,\bT),
\]
where $\partial = \partial/\partial T$. Since for any complex $z_i$
\[\sum_{a_1,\ldots,a_D=1}^N (X_z)_{a_1,\ldots,a_D} (\overline X_z)_{a_1,\ldots,a_D}\geq 0,
\text{ holds, then  }
\mathcal M (g)\succeq 0  \label{CNDproduct}\]
where the entries of this matrix are defined by (recall $\E=\E\hp N_{g}$)
\[
\mathcal M_{\alpha,\nu} (g): &=
  \E [
\bar\partial A_\alpha \cdot
\partial A_\nu ] {}\ \\
&=
\sum_{a_1,\ldots,a_D=1 }^N
 \E \bigg[
\frac{\partial A_\alpha(T,\bT) }{\partial \bT_{a_1,\ldots,a_D}} \times
\frac{\partial A_\nu(T,\bT) }{\partial T_{a_1,\ldots,a_D}}   \bigg]\,.
\]
Observe how $\partial A_\nu$ sums \textit{all} possible
terms that can be formed from
removing from $\partial A_\nu$ a white vertex in the sense of
Eqs. \eqref{bubble_removal} above (\textit{all} `open bubbles'\footnote{
\cite{Reiko_2BTM} exhibits a \psdm, whose entries are not built
by sums over `all open bubbles', but from single `open bubbles'. This matrix can be exploited
for graphs without symmetries and at finite-$N$. However, independently of those symmetries, for
melonic observables of melonic and non-melonic models, both matrices
introduce the same constraints at large-$N$, see Thm. \ref{thm:SDE_are_equal_thm}}); similarly,
$\bar\partial A_\alpha$ is the sum of all terms
that arise after removing black vertices from $A_\alpha$.
Then the inner product $\bar\partial A_\alpha \cdot
\partial A_\nu $ welds colour-wise all broken edges---which bear indices as in Eq. \eqref{bubble_removal}---after removing both
vertices.
\\[-1ex]

\noindent
This setting is able to accommodate in the  \psdm{} $\mathcal M$
any observable, melonic or not,  \textit{except the constant invariant}.
\\

\itemb
A second matrix, introduced in \cite{Reiko_2BTM}, is obtained by taking the expectation of
$\Tr Y_z^*Y_z\geq 0$ where
\[
Y_z(T,\bT)  &= \sum_\alpha z_\alpha   M _\alpha\hp c (T,\bT),\vspace{-2ex}
\]
and $M _\alpha\hp c$ is a coloured matrix constructed, paraphrasing the authors, as follows.
One chooses a colour  $c$ and for any invariant
$A_\alpha$ one picks an edge $e_c^\alpha$ of colour $c$ and cuts it.
The algebraic expression that remains $(A_\alpha\setminus \{e_c^\alpha\})_{a_c,b_c}$ is a matrix defined by the relation \vspace{-2ex} \[
\Tr (A_\alpha\setminus \{e_c^\alpha\}) =\sum_{a_c,b_c=1}^N  \delta_{a_c,b_c} ( A_\alpha\setminus \{e_c^\alpha\})_{a_c,b_c}= A_\alpha (T,\bT).\]
Set $M\hp c_\alpha= A_\alpha\setminus \{e_c^\alpha\}$ above. Repeating the procedure sketched for the matrix $\mathcal M$,
this leads to a \psdm{}
\[
 \mathcal N (g) \succeq 0, \qquad \with \mathcal N_{\alpha,\nu}(g) := 
 \E \big[ \Tr  \big(  [ M\hp c _\alpha]^* M\hp c _\nu \big) \big], \qquad  \text{ \cite{Reiko_2BTM}}\]
that
\textit{is able}
to accommodate the constant invariant and melonic invariants,
since, by definition, $ \delta_{a_c,b_c}$ contracted with
any of the vertex-excised invariants, yields an invariant (and contraction with itself yields $N$).

\end{itemize}

\subsection{Non-melonic boostrap \psdm s}
Although
 \cite{Reiko_2BTM} and \cite{TensorBootstrapProgram} are correct, both
 focused on melonic observables, and generalizations require some care regarding the scaling
 factors. Indeed,  $\mathcal M$ and $\mathcal N$
 are \psdm's, but the
 expectation  of some non-melonic observables scales as $1$ at large-$N$,
 while melonic graphs scale as $N$, and (after dividing the matrices by $N$) this
 mismatch of scaling factors suppresses non-melonic entries.  To make both matrices ready for non-melonicity,
 we scale them.  Let  $q_\alpha$ be such that\footnote{For instance
$q_{\alpha}=1,0,-1$ if $\alpha$ indexes a melon, $K_{3,3}$ or $K_{3,3} \cup_{x,y} K_{3,3}$, respectively, the latter being independent of $x$ and $y$.} 
$N^{-{q_\alpha}}\E [ A_\alpha ] \sim 1 $ at large-$N$, and let
 \[
X_z(T,\bT) & := \sum_\alpha z_\alpha N^{1/2-q_\alpha}   \partial A_\alpha(T,\bT) ,\\
Y_z(T,\bT) & := \sum_\alpha z_\alpha  N^{1/2-q_\alpha}    M _\alpha\hp c (T,\bT).
\]
Again the arbitrariness of the $z_\alpha$'s and
the non-negativity of the norms \eqref{CNDproduct} in $\C^D$ and $Y_z\mapsto \Tr(Y_z^*Y_z)$ in $M_N(\C)$, respectively,
lead to $\mathcal  M(g)\succeq 0$ and $\mathcal N(g)\succeq 0$ for
\begin{subequations}\label{PositiveMatrices_wNscalings}
\[
\mathcal M_{\alpha,\nu} (g): &= \label{PositiveMatrices_wNscalings_M}
N^{1- q_\alpha -q_\nu  } \E [
\bar\partial A_\alpha \cdot
\partial A_\nu ] ,  \\
\mathcal N_{\alpha,\nu} (g) :&=
N^{1- q_\alpha -q_\nu  }  \E \big[ \Tr  \big(  [ M\hp c _\alpha]^* M\hp c _\nu \big) \big]. \label{PositiveMatrices_wNscalings_N}
\]\end{subequations}%

\begin{proposition}\label{thm:NonMelonicBoostrapNscalings_Propo}
The two \psdm's \eqref{PositiveMatrices_wNscalings}
are finite and their non-melonic entries at large-$N$ are non-vanishing, whenever $N^{-q_{\alpha} } \E[A_\alpha]$ and $ N^{-q_{\nu} } \E[A_\nu]$ are finite and non-vanishing in that limit.
\end{proposition}

\begin{proof}[Sketch of proof.] To prove that the matrices with entries given by Eq.
\eqref{PositiveMatrices_wNscalings} are both \psdm's is routine.
To prove the interesting part of the claim, assume that $N$ is large.
First, we prove that $\E[
\bar\partial A_\alpha \cdot
\partial A_\nu]$ scales as $ N^{ q_\alpha  + q_\nu -1  } $, so that all entries in
\eqref{PositiveMatrices_wNscalings_M} are finite and not trivial, and leave
the analogous statement for \eqref{PositiveMatrices_wNscalings_N}, which has a similar proof,
sketched.\\

For any Feynman graph $H$, which contains always $0$-coloured edges (propagators),
let $H^{\hat 0}$ denote the graph remaining after removing all the $0$-coloured edges.
Let now $G_\alpha$ a connected maximal  Feynman graph
of the model  Eq. \eqref{Model} in which $A_\alpha$ sits
Wick-contracted, and let $\mathcal B_\alpha$ such that
 $G_\alpha ^{\hat 0} = \{A_\alpha\} \cup \mathcal B_\alpha$, which is the set of `interaction vertices'.
Since we are in the fixed model  Eq. \eqref{Model},
$ \mathcal B_\alpha = \{B_i\}_{i\in I(\alpha)}$
for some subset $I(\alpha)$ of $\{1,\ldots,m\}$ (following our previous
terminology let's call  $ \mathcal B_\alpha $ just bulk${{}_\alpha}$).
By definition of $q_\alpha$, $\E[A_\alpha]$ scales as $N^{q_\alpha}$,
then so does the cumulant or connected expectation $\E\conn[A_\alpha]$, and we can focus on connected Feynman graphs.
Since $A_\alpha$ is connected as well, and by assumption the maximum of faces $f_\alpha^{\max}= f(G_\alpha)$
is attained at $G_\alpha$, we have the following relations [recall Eqs.
 \eqref{Model}, \eqref{amp_def} and \eqref{ExpectDefinition}],
\[ \label{qf}
q_\alpha = f_\alpha^{\max} - (D-1) p_\alpha +   \sum_{i\in I(\alpha)} [ (D-1)p_i + s_i] ,
\]where
 $2p_\alpha=\#V(G_\alpha)$ and $p_i=p(B_i)$. We call the last sum over $i$
 bulk${}_{\alpha}$, abusing on  notation. 
\\

Now observe that the two derivatives $\bar \partial $  and $\partial$
with respect to $\bT$ and $T$, respectively, on each vertex yield
the following sum over the graph-unions of
excised vertices  \[
\bar\partial A_\alpha \cdot
\partial A_\nu
=\sum_{\substack{ x \in V_0(A_\alpha)\\ w \in V_1(A_\nu)} }
b_{x,w}^{\alpha,\nu}
\qquad \with\qquad
b_{x,w}^{\alpha,\nu}: =
(A_\alpha)
\underset{x,w}{\cup}
(A_\nu). \label{bxw}
\]

It is not true that any maximal graph of the Model \eqref{Model} containing
$ b_{x,w}^{\alpha,\nu}$ splits as two maximal graphs, one containing
of $A_\alpha + $\,bulk${}_{\alpha}$ and $A_\nu +$\,bulk${}_{\nu}$. However, as in
Lemma \ref{thm:union_Faces_sum_lemma}, the latter two sets inject into the former,
or in other words: although we cannot split any $H$ that attains
 $f^{\max}_{\alpha \cup \nu;x,w} = f(H)$ into two Wick contractions---one
 of $A_\alpha$ (and bulk) and $A_\nu$ (and bulk)---we can design a maximal
 $H$ from maximal data for $A_\alpha$ and $A_\nu$. Namely, first let
 $\text{bulk}_{\alpha\cup\nu}  = \sum_{I(\alpha) \cup I(\nu)}  (D-1)p_i + s_i$
 and let us now care about the Wick contraction.\\

 Let $\pi_\alpha$ be the Wick contraction underlying $G_\alpha$, i.e.
$G_\alpha=\pi_\alpha (A_\alpha \dot\cup \mathcal B_\alpha )$ (and similarly for $\nu$)
and define $\Pi_{\alpha,\nu;x,w} =
(\pi_\alpha \setminus \{x,y_\alpha \} )\cup (\pi _\nu \setminus \{v,w_\nu\}  )\cup \{ (v,y) \}
$, where the new vertices are uniquely determined by $(x,y_\alpha) \in \pi_\alpha $
and $ (v,w_\nu) \in \pi_\nu$. By Eq. \eqref{faces_deficit_union},
this Wick contraction $\Pi_{\alpha,\nu;x,w} $ is maximal if $\pi_\alpha$ and $\pi_\nu$ are. Moreover, observe
that the maximum number of faces of $\Pi_{\alpha,\nu;x,w} $ reads \vspace{-2ex}
\[ \label{sumeoffacesinproof}
f^{\max}_{\alpha \cup \nu ; x,w} = f^{\max}_\alpha + f^{\max}_\nu - D,
\]
which is independent of $x$ and $w$. Notice that
$ p \big( b_{x,w}^{\alpha,\nu} \big)$ is half the vertices of $b_{x,w}^{\alpha,\nu}$, so
$p \big( b_{x,w}^{\alpha,\nu} \big) =p_\alpha+p_\nu-1$, which is independent of $x,w$.
Let $q(\alpha \cup \nu)$ denote the large-$N$ scaling exponent
of $\bar\partial A_\alpha \cdot
\partial A_\nu$, $\E\conn [ \bar\partial A_\alpha \cdot
\partial A_\nu] \sim N^{q(\alpha\cup \nu) }$.
Then if $f^{\max}_{\alpha \cup \nu;x,w}$ is the maximum of faces
$b_{x,w}^{\alpha,\nu}$ can receive in the Model \eqref{Model},
\[q(\alpha \cup \nu)\notag
&=   f^{\max}_{\alpha \cup \nu;x,w} -  (D-1)  \times p \big( b_{x,w}^{\alpha,\nu} \big)+  \text{bulk}_{\alpha\cup\nu}  \\
 & = [f^{\max}_\alpha + f^{\max}_\nu  - D] -  (D-1) \times (p_\alpha+p_\nu-1 ) +\textstyle  \sum_{I(\alpha)\cup I(\nu)} \notag
(D-1)p_i + s_i \notag\\ \notag
&= f_\alpha^{\max} - (D-1) p_\alpha + \textstyle  \sum_{i\in I(\alpha)} [ (D-1)p_i + s_i] \\
& +\,f_\nu^{\max} - (D-1) p_\nu + \textstyle  \sum_{i\in I(\nu)} [ (D-1)p_i + s_i]  -1 \notag \\
& =  q_\alpha + q_\nu  -1\notag .
\]
To prove that the entries of the second matrix
\eqref{PositiveMatrices_wNscalings_N} also scale as $N^{q_\alpha + q_\nu  -1}$,
one can use Eq. \eqref{swap_faces_deficit},
which implies that the way one traces the product
$M_\alpha^* M_\nu$ (which form a sort of edge-swap, but
with edges $c$-coloured with $c>0$) one obtains the exact
faces deficit  as in \eqref{sumeoffacesinproof}
with respect to the faces of $\Tr M_\alpha$
and $\Tr M_\nu$. (Alternatively, since
this swap is additive in the Gur\u au-degree, which we did not use here,
\cite[Prop. 3.11]{fullward} can be used to write a shorter proof.)
\end{proof}
\begin{table}\begin{minipage}{.70\textwidth}
\includegraphics[width=.996\textwidth]{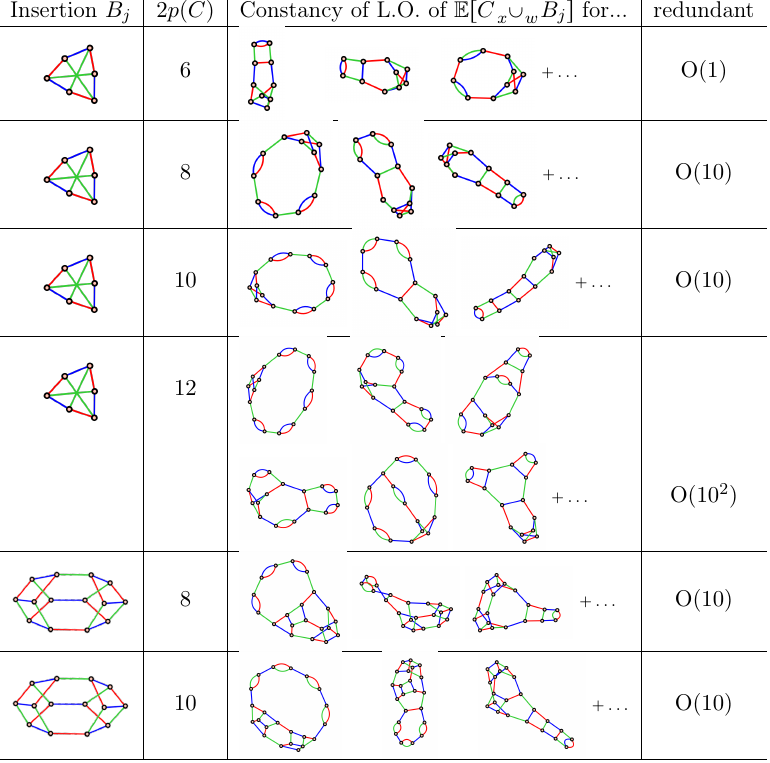}\hspace{-2ex}
             \end{minipage}\hspace{-5ex}\begin{minipage}{.33\textwidth}
\caption{Theorem \ref{thm:SDE_are_equal_thm} states that in
the large-$N$ moment of the union of a graph
$B_j$ with a melon $C$ one can replace $C$ by
$\tilde C$ as far as $p(C)=p(\tilde C)$. This simplifies
tensor bootstraps by treating all listed invariants in each row
one moment. \label{tab:nonBoteros}} \end{minipage}
\end{table}%

\allowdisplaybreaks[4]
\section{Conclusion}\label{sec:concl_outlook}

By ordering labels of vertex-labellings of melons,
we are able to locate dipole insertion vertex pairs
and replace a melon (which sits Wick-contracted in a larger 
graph) by any other melon with the same number of vertices, 
without modifying its maximum of faces.
This leads to the independence of the
combinatorial details of large-$N$ melonic moments in
non-melonic tensor models; instead, they only datum 
they depend  on is their number of vertices.
A similar result
holds for the `melonic part' (Thm. \ref{thm:SDE_are_equal_thm}) in
any graph, which leads to the equality of the Schwinger-Dyson Equations
for melonic observables of the same number of vertices. 
However, the second type of universality  
of \cite{TwofoldUniversality}, to wit \textit{independence of $D$}
in melonic large-$N$ models, does not seem to have a counterpart 
for non-melonic models, unsurprisingly (this is not 
disappointing, it is rather the
universality in $D$ of \cite{TwofoldUniversality} which is surprising).\\

An application of our result is the reduction of
the number of independent moments in positivity
bootstraps. The way invariants scale, depending on how we cut them and build
 a positive semidefinite matrix seems to be a criterion that deserves study.
Indeed, consider for instance the basis of observables
depicted on the top of the following matrix:
\[ \nonumber
 \tfrac 1N \runter{\includegraphics[width=9ex]{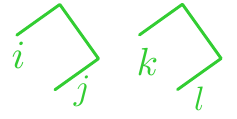}}
    \tfrac 1N \runter{\includegraphics[width=9ex]{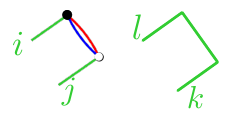}}
 \runter{\includegraphics[width=12ex]{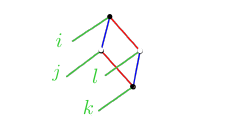}}       \tfrac 1N \runter{\includegraphics[width=9ex]{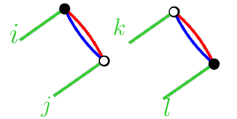} } \\[-3ex]
\kbordermatrix{ \cr
\cr
   \tfrac 1N \runter{\includegraphics[width=10ex]{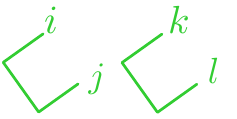} }
 &  1
  &   \tfrac 1N  \,\,\raisebox{-.75ex}{\includegraphics[width=3ex]{melon_tricolor} }
   &   \tfrac 1N  \runter{\includegraphics[width=5ex]{V2colors} } &  \tfrac 1N \runter{\includegraphics[width=5ex]{V1colors} }
  \cr
 \tfrac 1N   \runter{\includegraphics[width=10ex]{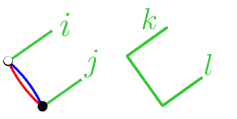}}
   &   \tfrac 1N  \,\,\raisebox{-.75ex}{\includegraphics[width=3ex]{melon_tricolor} }
  & \tfrac 1N  \runter{\includegraphics[width=5ex]{V1colors} }
& \tfrac 1N \runter{\includegraphics[width=5ex]{E1colors} }& \raisebox{-.75ex}{\includegraphics[width=3ex]{melon_tricolor} }\hspace{-1.15ex}\runter{ \includegraphics[width=5ex]{V1colors} }
  \cr
   \runter{\includegraphics[width=13ex]{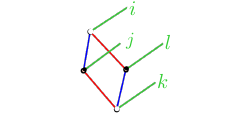} }
   &   \tfrac 1N  \runter{\includegraphics[width=5ex]{V2colors} }
  &  \tfrac 1N  \runter{\includegraphics[width=5ex]{E1colors} }
  &
  \runter{\includegraphics[width=7ex]{Cubecolors} }
   &  \tfrac 1N \runter{\includegraphics[width=6ex]{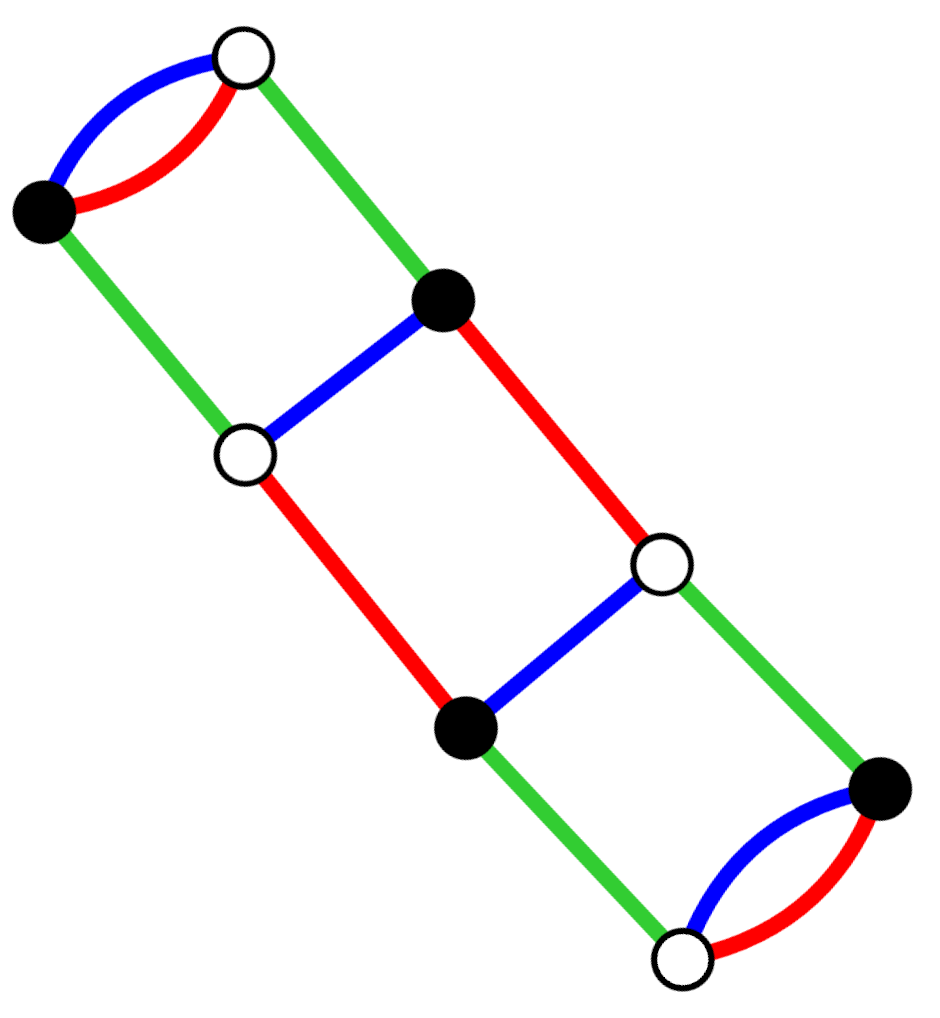} }
  \cr
  \tfrac 1N \runter{\includegraphics[width=10ex]{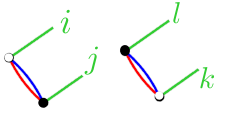}}
   &  \tfrac 1N \runter{\includegraphics[width=5ex]{V1colors} }
  &  \raisebox{-.75ex}{\includegraphics[width=3ex]{melon_tricolor} }\hspace{-1.15ex}\runter{\includegraphics[width=5ex]{V1colors} } &
  \tfrac 1N \runter{\includegraphics[width=6ex]{Mel121} }  &
  \runter{\includegraphics[width=5ex]{V1colors} }\!\!\!\!\!\runter{\includegraphics[width=5ex]{V1colors} }
} \nonumber
\]
The corresponding matrix of moments (i) is positive semidefinite, (ii)  scales correctly at large-$N$ and
(iii) does not arise from the available techniques (Sec. \ref{sec:PosSemiDef_Tensors}).
Indeed, a random tensors allow a 
diversity in the construction of positive semidefinite matrices.
This direction can be exploited as an alternative to---or at least an
enrichment of---large matrices of  moments.
Moreover, the \textit{non-melonic moments} that appear as entries of the
two positive semidefinite matrices that so far
build the core of the tensor
bootstrap (Eq. \eqref{PositiveMatrices_wNscalings},
which essentially correspond to the non-melonic adaptation of
\cite{TensorBootstrapProgram}'s and \cite{Reiko2026}'s idea),
scale now correctly at large-$N$
and are ready for non-melonic bootstrap,
as proven in Proposition \ref{thm:NonMelonicBoostrapNscalings_Propo}.
In any case, Tables \ref{tab:Simplifyning}
and \ref{tab:nonBoteros} show the magnitude of the simplification that melonic moments of
arbitrary tensor models will enjoy in the large-$N$
thanks to our main results (Prop. \ref{thm:ExpMelonEqualPropo} and Thm. \ref{thm:SDE_are_equal_thm}),
independently of their utility for bootstraps.

\vspace{1cm}

\appendix
\subsection*{Acknowledgements}
Feynman graphs and almost all $D$-coloured graphs were drawn using \texttt{feyntensor} \cite{feyntensor}, written by the author in---and thanks to---SageMath \cite{sagemath}.

\section*{Table of notation}
We list the meaning of some frequently used symbols. \\
\begin{center}
\fontsize{10.75}{13.55}\selectfont
 \begin{tabularx}{.985\textwidth}{lX}
$\mathcal M \succeq 0$ & positive semidefinite $\mathcal M$\\
 $\dot\cup$ & disjoint union\\
 $\#$ & cardinality \\
$*$ & adjoint (elsewhere dagger) \\
  $G\swap{v}{w} H$ & swap of $G$ and $H$ at $v$ and $w$  \\
$\cup_{x,y}$ & union of two graphs at two vertices $x,y$\\
$(r_{d} r_{d-1} \cdots r_1 r_0)_D$  & base-$D$ notation, $\sum_{i=0}^d r_i D^i$\\
 $ \mathcal A$ & arborescence (of a melon) \\
$B,C,\tilde B,\tilde C, B_i$ & typically, $D$-coloured graphs\\
$B, \tilde B, B_i$ & interaction vertices, in general non-melonic\\
$C,\tilde C $ & observables (typically melonic) \\
$C(T,\bT)$ & invariants in the tensor $T$ defined by the graph $C$ \\
$D$ &  number of colours of $T \in (\C^N)^{\otimes D}$, `rank of $T$' \\
$D$-coloured graph & vertex-bipartite uniformly edge-$D$-coloured graph
\\
$\E[B] $ & moment/expectation value, $\E_{g_1,\ldots,g_m} \hp {N}[B(T,\bT)]$ explicitly \\
\dip & dipole insertion vertex pair \\
$F(G)$ & faces of a Feynman graph\\
$f(G)$ & number of faces of a Feynman graph\\
$g_1,\ldots, g_m$ & coupling constants/parameters of the measure\\
$g$ & abbreviates the list of couplings $ g_1,\ldots, g_m$\\
$G$, $H$ & Feynman graphs\\
L.O. & Leading Order in $N$ \\
$m$ & number of interactions in $S\inter$ \\
$\mathscr M_n$ & face-maximizing Wick contractions (of the $n$ arguments)\\
$n$ & order in perturbation theory \\
$N$ & size of $N$, $T \in (\C^N)^{\otimes D}$\\
$p(C)$ & half of the number of vertices of $C$\\
\psdm & positive semidefinite matrix \\
$\pi$ & Wick contraction \\
$\pi(B)$ & Wick contraction of $B$, a $(D+1)$-coloured Feynman graph\\
$S\inter(T,\bT)$ & $\sum_{j=1}^m g_j N^{s_j} B_j(T,\bT) $  \\
$\Sym(p),\Sym(\Omega)$ & bijections/permutations of $\{1,\ldots,p\}$, resp. on a set $\Omega$    \\
$T$ & tensor \\
$\mathcal T$ & tree of a melon \\
$V_0(B)$ & black or even vertices of $B$ \\
$V_1(B)$ & white or odd vertices of $B$ \\
$V(B)$ & $V_0(B)\dot\cup V_1(B)$ \\
$v,w,x,y$ & typical vertex variables\\
$\bar z$ & complex conjugate of $z$ \\
$\mathcal Z$ & partition function
\end{tabularx}

\end{center}

\newgeometry{left=1.84cm,right=1.84cm, bottom=2.5cm}
{
\bibliographystyle{alpha}

}
\end{document}